\documentclass[a4paper]{lipics-v2021}
\usepackage{amsmath,amssymb}
\usepackage{graphicx}
\usepackage{pdflscape}
\usepackage{algpseudocode}
\usepackage{booktabs}
\usepackage{enumitem}
\usepackage{hyperref}
\usepackage{xcolor}
\usepackage{amsthm}
\usepackage{rotating}
\newcommand{\best}[1]{\textcolor{green!50!black}{\textbf{#1}}}
\newcommand{\second}[1]{\textcolor{orange!70!black}{#1}}

\usepackage[ruled,vlined, linesnumbered]{algorithm2e}

\SetCommentSty{mycommfont}
\SetAlFnt{\footnotesize}

\newcommand{\bigO}{\mathcal{O}}
\newcommand{\task}{\tau}
\newcommand{\job}[2]{\task_{#1,#2}}
\newcommand{\jld}[4]{\task_{#1} \xrightarrow{(#3,#4)} \task_{#2}}
\newcommand{\chain}{\zeta}
\newcommand{\hp}{\mathit{HP}}
\newcommand{\wcet}{C}
\newcommand{\rmin}{R_{\min}}
\newcommand{\rmax}{R_{\max}}
\newcommand{\dmin}{D_{\min}}
\newcommand{\dmax}{D_{\max}}
\newcommand{\dpmin}{D'_{\min}}
\newcommand{\WCL}{\mathit{WCL}}

\title{Schedulable Job-Level Dependencies for Cause-Effect Chains via Graph Neural Networks}

\titlerunning{Schedulable Job-Level Dependencies for Cause-Effect Chains via GNN}

\author{Silviu S. Craciunas}{NXP Semiconductors, 
    Vienna, Austria\\Technical University of Denmark, 
    Kongens Lyngby, Denmark}{silviu.craciunas@nxp.com}{https://orcid.org/0000-0003-0116-9265}{}

\author{Christian Hakert}{TU Dortmund, Germany}{christian.hakert@tu-dortmund.de}{https://orcid.org/0000-0001-9992-9415}{}    

\author{Jian-Jia Chen}{RWTH Aachen, Germany}{jian-jia.chen@rwth-aachen.de}{https://orcid.org/0000-0001-8114-9760}{}    

\author{Zdeněk Hanzálek}{Czech Technical University, Czech Republic}{zdenek.hanzalek@cvut.cz}{https://orcid.org/0000-0002-8135-1296}{}    

\author{Paul Pop}{Technical University of Denmark, 
    Kongens Lyngby, Denmark}{paupo@dtu.dk}{https://orcid.org/0000-0001-9981-1775}{}

\authorrunning{S. S. Craciunas et al.} 

\Copyright{S. S. Craciunas, C. Hakert, J-J. Chen, Z. Hanzálek, P. Pop}


\category{} 

\relatedversion{} 
\hideLIPIcs  

\EventEditors{}
\EventNoEds{2}
\EventLongTitle{}
\EventShortTitle{}
\EventAcronym{}
\EventYear{}
\EventDate{}
\EventLocation{}
\EventLogo{}
\SeriesVolume{}
\ArticleNo{}

\begin{document}
\nolinenumbers

\maketitle

\begin{abstract}	
Modern automotive software architectures comprise large sets of mixed-criticality functions executing on shared multi-core platforms with strict real-time and end-to-end timing requirements. Sensor-to-actuator data propagation in such systems is typically expressed via cause-effect chains with worst-case data-age budgets. Job-level dependencies (JLDs) have been introduced to provide a schedule-agnostic mechanism for bounding the data age independently of the underlying scheduler. The state-of-the-art methods for synthesizing JLDs, however, do not check whether the produced JLDs are enforceable under a concrete scheduling policy or jointly schedulable at the system level. In this paper we propose the first machine-learning-based JLD synthesis method, built around a two-level Graph Neural Network with temperature-controlled sampling that learns the structural patterns connecting cause-effect chain configurations to their JLD solutions. Since learned outputs may not be correct by construction, we embed the GNN in a novel Generate-and-Verify architecture in which a  safe DP data-age checker, together with a per-chain EDF feasibility checker and a system-level demand-bound test, accept or reject each candidate.  We show that the ML-based generator substantially outperforms the original greedy heuristic  while achieving orders-of-magnitude lower synthesis time, demonstrating that learned structural priors can effectively replace exponential propagation-tree enumeration on this class of real-time scheduling problems.
\end{abstract}

\section{Introduction}
\label{sec:intro}
\sloppypar
Modern automotive systems are shifting from a decentralized ECU architecture to a more centralized and zonal software-defined approach in which hundreds of mixed-criticality functions, like Advanced Driver Assistance Systems (ADAS), Assisted Driving (AD), infotainment, etc., run on a shared high-performance multi-core multi-SoC platform~\cite{niedrist, mclean2020mapping}. In such integrated automotive platforms, ADAS functions must satisfy not only strict real-time deadline requirements but also guarantee complex dependencies among sensors, control software, and actuators. A typical driver-assistance function, for example, gathers data from cameras and distance sensors (e.g., ultrasonic sensors and LiDAR), combines them in a sensor-fusion layer, and passes the resulting information to navigation and control algorithms that generate actuator commands (e.g., emergency braking). This execution sequence forms a temporal dependency chain that must respect timing constraints on the freshness of the data. In automotive systems, such dependencies are commonly modeled as \emph{cause-effect chains}, that is, sequences of periodic tasks through which data propagates from a sensor, across multiple processing tasks, to an actuator~\cite{becker2016synthesizing, becker2018endtoend}. Each chain is subject to a \emph{data age} constraint, which bounds the maximum latency between sampling at the source and the point at which the actuator output reflects that input. In multi-rate chains, where tasks execute at different frequencies, the worst-case data age depends on the particular job instances that communicate, making the analysis non-trivial~\cite{grus2025packing, grus2025chains}. Ensuring both functional interdependence and compliance with these end-to-end timing constraints therefore requires careful scheduling of ADAS functions.

Becker et al. proposed a schedule-agnostic solution for enforcing chain data-age constraints in the form of \emph{Job-Level Dependencies} (JLDs)~\cite{becker2016synthesizing}. As shown in~\cite{becker2016synthesizing}, a consumer reads whichever value is currently in the shared register under implicit communication (``last-is-best''). Without constraints, the worst-case data age can far exceed the budget. A JLD $\jld{i}{j}{k}{l}$ enforces that job~$k$ of $\task_i$ must complete before job~$l$ of $\task_j$ begins. This prunes paths in the data propagation tree that would cause excessive latency, bounding the worst-case data age independently of the concrete schedule dispatch. Hence, given a set of correct JLDs, a scheduler (either an offline created TT schedule~\cite{FinziRTAS24} or a runtime algorithm such as EDF/FP~\cite{liu73}) only need to enforce the JLDs to make the schedule correct in terms of chain data ages. While JLDs are effective for data age reduction, they have two fundamental limitations. First, synthesizing the JLD set requires exponential time in chain length due to full propagation tree enumeration~\cite{becker2016synthesizing, 10.1145/3177803.3177805}. Secondly, the Greedy-JLD heuristic by Becker et~al.~\cite{becker2016synthesizing} does not check whether the produced JLDs can be \emph{enforced by a scheduler}. Each JLD required for satisfying the data age constraint of a chain may lead to either an unenforceable constraint or an infeasible schedule. 

In this paper, we address the problem of synthesizing enforceable and schedulable JLDs for cause-effect chains in EDF-scheduled systems. We propose a \emph{Generate-and-Verify} architecture that integrates JLD candidate generation with verification in a synthesis loop, as well as two novel methods for generating JLDs. We propose a randomized variant of the Greedy-JLD heuristic which allows multiple JLDs per pair, called \textit{GreedyRand-JLD} and a ML-based generator with temperature-controlled sampling called \textit{GNN-JLD}. The intuition behind the ML-based variant is that cause-effect chains contain patterns related to period structure, relations and transitions, as well as utilization and chain structure to allow a learning-based approach to gain meaningful information that allows an efficient generation of JLDs even for unseen cases. Since the ML-generated JLDs are not guaranteed to be correct, the generate-and-verify architecture makes sure that the JLDs result in correct data-ages and are also EDF-enforceable and the resulting system schedulable under EDF. We evaluate both the GreedyRand-JLD and GNN-JLD against a random generator baseline and the Greedy-JLD heuristic without the generate-and-verify loop on 1,500 use-cases spanning use-cases from the training set, similar use-cases as in the training set, and unknown use-cases in terms of period sets. Both methods substantially outperform the original Greedy-JLD heuristic. GreedyRand-JLD achieves $80\%$ fully valid use-cases compared to $64\%$ for single-attempt Greedy-JLD, while the GNN-based generator further raises this to $89\%$ at a median per-chain cost of $6$--$8$\,ms ($3$--$5\times$ faster than heuristics at the median and over $100\times$ faster in the tail).

In Sec.~\ref{sec:background} we present the system model and data age analysis from~\cite{becker2016synthesizing} and describe the EDF schedulability enforcement in Sec.~\ref{sec:enforcement}. Sec.~\ref{sec:generate_and_verify} introduces the Generate-and-Verify architecture, the DP data age checker, and the GreedyRand-JLD generator.
Sec.~\ref{sec:gnn} presents the ML-based approach, including features sets and training method. We evaluate the approaches in Sec.~\ref{sec:evaluation}, discuss related work in  
Sec.~\ref{sec:related} and conclude in 
Sec.~\ref{sec:conclusion}.

\section{Background and System Model}
\label{sec:background}

\subsection{Task and Chain Model}
\label{sec:task-model-chain-model}
We consider a set of periodic tasks $\Gamma = \{\task_1, \ldots, \task_n\}$ on a single processor under preemptive EDF. We note that in modern multi-core automotive systems, safety critical functions are usually pre-assigned to cores and schedulability is handled in a fully partitioned manner. The task-to-core assignment has been studied for automotive systems~\cite{mclean2020mapping} and is not in scope of this work. We also note that a (semi-)partitioned approach can achieve near-optimal schedulability compared to a global scheduling approach~\cite{7809847}. 

Each task $\task_i$ is defined by a period $T_i$ and a WCET $\wcet_i$. In this work we assume implicit deadlines to be able to compare to Becker et~al.~\cite{becker2016synthesizing}\footnote{We note that supporting constrained-deadlines ($D_i \le T_i$) is trivial and already integrated in our ML model (via the deadline ratio feature) but the current training data and benchmark focus on the implicit deadline case in order to be able to faithfully compare to~\cite{becker2016synthesizing}.}. The $j$-th job of $\task_i$, denoted $\job{i}{j}$, is therefore released at $(j{-}1) T_i$ and must complete by $j \cdot T_i$. We denote the hyperperiod with $\hp = \mathrm{lcm}(T_1, \ldots, T_n)$.

A \emph{cause-effect chain} $\chain = \task_{s_1} \to \task_{s_2} \to \cdots \to \task_{s_m}$ is a sequence of tasks through which data propagates. Each chain has a maximum data age requirement~$B$. In~\cite{becker2016synthesizing, becker2018endtoend} four timing bounds (for implicit deadlines) are defined for each job $\job{i}{j}$:
{\small\begin{align}
    \rmin(\job{i}{j}) &= (j{-}1) \cdot T_i &
    \rmax(\job{i}{j}) &= \rmin(\job{i}{j{+}1}) - \wcet_i \nonumber \\
    \dmin(\job{i}{j}) &= \rmin(\job{i}{j}) + \wcet_i &
    \dmax(\job{i}{j}) &= \rmax(\job{i}{j{+}1}) + \wcet_i \nonumber
\end{align}}

The \emph{read interval} $[\rmin,\rmax]$ is the window during which a job may begin reading its input data. The \emph{data interval} $[\dmin,\dmax]$ describes the time during which the job's output may be visible before it is overwritten by the next instance. The bounds from~\cite{becker2016synthesizing,becker2018endtoend} assume delayed publication no earlier than $\rmin+\wcet$, which is sound under enforced delayed output or fixed execution time (instead of WCET) of every periodic task. With immediate publication and executions shorter than WCET, $\dmin=\rmin+\wcet$ is optimistic and may miss feasible propagation paths\footnote{We note that our method is independent of the publication semantic and works with $\dmin=\rmin$ or $\dmin=\rmin + BCET$ or with LET semantics~\cite{Kirsch2012}. Expanding the method to such semantics is trivial and just modifies the $\dmin$ formula.}. A job $ \job{k}{l}$ can consume data produced by $\job{i}{j}$ iff their read and data intervals overlap. For chains of length $>2$, the data interval requires a forward adjustment~\cite{becker2018endtoend}:
\begin{equation}
\dmin'(\job{k}{l}, \job{i}{j}) = \max(\dmin(\job{k}{l}),\; \dmin(\job{i}{j}) + \wcet_k)
\label{eq:dpmin}
\end{equation}
Note that Eq.~\ref{eq:dpmin} is applied iteratively along a path: at each level the carried $\dpmin$ replaces $\dmin(\job{i}{j})$ in the formula, so that the cascading correctly accumulates over multiple hops.

The function $\text{maxAge}(\task_{root}, \task_{leaf}, dpt)$ returns the maximum latency of the path from $\task_{root}$ to $\task_{leaf}$, where both root and leaf job are part of the data propagation tree $dpt$~\cite{becker2018endtoend}:
$\text{maxAge}(\task_{root}, \task_{leaf}, dpt) = (\rmax(\task_{leaf}) + \wcet_{leaf}) - \rmin(\task_{root}).$ The global maximum data age of a cause-effect chain, denoted $l_{\max}$, considers all root/leaf combinations across all data propagation trees in $\mathcal{T}$~\cite{becker2018endtoend, becker2016synthesizing}:
  \begin{equation}
      l_{\max} = \max_{\forall\, \task_{root}, \task_{leaf} \in dpt,\; \forall\, dpt \in \mathcal{T}} \text{maxAge}(\task_{root}, \task_{leaf}, dpt) \nonumber
  \end{equation}
where $\mathcal{T}$ is the set of all data propagation trees of chain~$\chain$, one rooted at each job of the first task within one hyperperiod.

\subsection{Greedy-JLD Heuristic and Its Limitations}
\label{sec:becker}
The Greedy-JLD heuristic (Algorithm~3 of~\cite{becker2016synthesizing}) iteratively builds the data propagation tree, finds an invalid path (data age $> B$), and adds one JLD to prune the path. The MECHAniSer reference implementation~\cite{becker2016mechaniser} simplifies the JLD-placement step to a free-pair scan. i.e., it walks the offending path from leaf to root and inserts the JLD at the first $(pred_{task}, succ_{task})$ pair that does not already carry a constraint, exploiting the algorithm's invariant of at most one JLD per task pair. If every pair on the path already carries a JLD, the algorithm cannot add another and reports the chain as over-restricted.

The generation cost is exponential in chain length: the search space has $\prod_{i=1}^{m-1} b_i$ candidate paths, where $b_i$ is the branching factor at edge $i$, bounded by the number of successor jobs whose intervals overlap a given predecessor's data window~\cite{becker2018endtoend}. MECHAniSer uses a greedy depth-first search with backtracking rather than full tree enumeration, but the worst-case complexity remains exponential. Moreover, the one-JLD-per-pair restriction limits how finely the algorithm can prune, i.e.,  each JLD blocks an entire subtree, and on longer or non-harmonic chains the few coarse-grained cuts can disconnect the chain entirely (no propagation path survives), eliminating valid latencies before the budget is met. 

Another significant limitation of the Greedy-JLD heuristic~\cite{becker2016synthesizing} is that it does not check enforceability under a specific scheduling mechanism (e.g.\ FP or EDF), nor does it perform a task set transformation that would enforce the JLDs under the specific scheduling algorithm, and it also does not check system-wide schedulability of the transformed task set under the given scheduling discipline.

\section{Schedule Enforceability}
\label{sec:enforcement}
In this section we show how to enforce JLDs in an EDF-scheduled system via task-set transformation and how to check the system-level feasibility using the transformed task set.

A JLD set that satisfies the data-age is not necessarily schedulable. Under EDF, enforcing a job precedence $A \to B$ requires that (i) $\phi_A \leq \phi_B$, i.e., the predecessor is released no later than the successor, and (ii) $\mathrm{dl}(A) < \mathrm{dl}(B)$, i.e., the predecessor has a strictly smaller absolute deadline, ensuring that EDF dispatches it first when both are pending.

We use a transformation similar to~\cite{338100, 896010, CraciunasETFA14}. Each task $\task_i$ involved in any JLD is replaced by $\hp / T_i$ virtual tasks, one for each job instance in the system hyperperiod $\hp$. The $j$-th virtual task $\task_{i,j}^v$ is defined by the tuple $\langle \phi_{i,j}, D_i, \wcet_i, \hp \rangle$, where $\phi_{i,j} = (j{-}1) T_i$ is the offset, $D_i$ is the relative deadline, $\wcet_i$ is the WCET, and the period is set to $\hp$. Tasks not involved in any JLD retain their original parameters. Note that the total utilization per original task is preserved by construction.

For each JLD pair $(A, B)$ between virtual tasks, we adapt the precedence enforcement rules from~\cite{chetto1990dynamic} to the case of periodic virtual tasks. Depending on the offset and deadline relationship between $A$ and $B$, four cases arise:
\begin{itemize}[leftmargin=2em]
\item[C.1] $\phi_A \leq \phi_B$ and $\mathrm{dl}(A) < \mathrm{dl}(B)$: the constraint is already satisfied; no modification is needed.
\item[C.2] $\phi_A \leq \phi_B$ and $\mathrm{dl}(A) \geq \mathrm{dl}(B)$: we pull $\mathrm{dl}(A)$ to $\mathrm{dl}(B){-}1$ or push $\mathrm{dl}(B)$ to $\mathrm{dl}(A){+}1$. We prefer pulling the predecessor deadline as it is less disruptive to subsequent successors; if doing so would violate $D'_A \geq \wcet_A$, we push the successor deadline instead.
\item[C.3] $\phi_A > \phi_B$ and $\mathrm{dl}(A) < \mathrm{dl}(B)$: we push $\phi_B$ to $\phi_A$, preserving the original absolute deadline of~$B$.
\item[C.4] $\phi_A > \phi_B$ and $\mathrm{dl}(A) \geq \mathrm{dl}(B)$: we first push $\phi_B$ to $\phi_A$, and then resolve the deadline conflict as in case C.2.
\end{itemize}

Every modification must keep the absolute deadline within its original release-to-deadline window and respect $D' \geq \wcet$, otherwise the JLD is wrong. When multiple JLDs form chains of precedences (e.g., $A \to B \to C$), we build a DAG over all dependent virtual-task pairs, topologically sort it, and apply the four cases in forward order (offset push followed by deadline ordering) and a backward pass then tightens the deadlines as needed. Cycles in the dependency DAG are detected during the topological sort and cause rejection of the JLD set.

Note that a JLD between two virtual tasks whose execution windows do not overlap is either \emph{trivially satisfied}, when the successor releases after the predecessor deadline, or \emph{impossible}, when the successor must finish before the predecessor can start. We discard the former (no constraint needs to be enforced) and reject the latter (the JLD set is infeasible).

After the transformation, system-level EDF schedulability is checked over the \emph{complete} task set (virtual tasks + non-JLD chain tasks + tasks that are not in any chain) by the demand-bound schedulability test from~\cite{baruah1990algorithms, pellizzoni2005feasibility}:
\begin{align}
    & \forall t_1 {\in} \Lambda,\; \forall t_2 {\in} \Delta,\; t_1 < t_2: \nonumber \\
    & \sum_{k=1}^{n} \wcet_k \!\left(\!\left\lfloor \frac{t_2 {-} \phi_k {-} D_k}{T_k} \right\rfloor {-} \left\lceil \frac{t_1 {-} \phi_k}{T_k} \right\rceil {+} 1\!\right)_{\!0} \!\leq t_2 {-} t_1
\end{align}
where $(x)_0 = \max(0, x)$, and the test point sets are defined up to $\Phi + 2 \cdot \hp$ with $\Phi = \max_i(\phi_i)$:
\begin{align}
    \Lambda &= \{\phi_i + j T_i \mid i{=}1,\ldots,n;\; j {\geq} 0;\; \phi_i {+} j T_i \leq \Phi + 2 \hp\} \nonumber \\
    \Delta  &= \{a_{i,j} + D_i \mid a_{i,j} \in \Lambda;\; a_{i,j} {+} D_i \leq \Phi + 2 \hp\} \nonumber
\end{align}

\section{Generate-and-Verify}
\label{sec:generate_and_verify}
We address the problem of synthesizing enforceable and schedulable JLDs for cause-effect chains in EDF-scheduled systems via a novel \emph{Generate-and-Verify} architecture that combines JLD candidate generation with verification within a synthesis loop (cf.\ Algorithm~\ref{alg:gen-verify}). The key idea is to decouple candidate generation from correctness checking: a \emph{candidate generator} $\mathcal{G}$ proposes a JLD set for each chain, and each candidate is then passed through a two-tier verifier consisting of (i) the data-age checker (defined below), which verifies that the worst-case age satisfies the chain's budget, and (ii) the per-chain EDF feasibility checker, which verifies that the resulting precedence constraints are physically enforceable under EDF using the transformation described in Section~\ref{sec:enforcement}.

A candidate is accepted only if both checks pass and otherwise the generator produces the next candidate. The loop runs up to $K$ attempts per chain. If a chain produces a passing candidate, that candidate's JLD set is stored for the system-level check. If all $K$ attempts fail (no candidate satisfies DA + per-chain EDF in isolation), the loop retains the \emph{best-effort} candidate seen (ranked first by verification tier DA + EDF $\succ$ DA-only $\succ$ neither, then by lower max-age, then by fewer JLDs). The chain is still marked unsatisfied in isolation, but its best-effort JLDs participate in the cross-chain analysis below. All methods follow the same accounting to ensure fair comparisons across methods.

After all chains have a candidate (passing or best-effort), a post-combine data-age recheck verifies each chain against the combined per-method union: cross-chain JLDs touching tasks shared between chains tighten adjusted intervals on those tasks (push $R_{\min}$ forward, pull $R_{\max}$ back), so an isolated DA failure can become satisfied under the combined set, and an isolated DA pass can in principle be invalidated. Then, the task set is transformed and the \emph{system-level} EDF demand bound test checks whether the combined JLD sets from all chains are jointly schedulable on a single core (cf. Section~\ref{sec:enforcement}). If this test fails, the entire per-chain loop is restarted with modified generator parameters for up to $S$ system retries. One valid candidate must be produced within K attempts for at least one of the S generator settings. This separation allows fast, approximate generators to be paired with verifiers, combining speed with correctness guarantees.

\begin{algorithm}[t]
\caption{Generate-and-Verify}
\label{alg:gen-verify}
\KwData{$c$ chains $\chain_1, ..., \chain_c$, tasks $\Gamma$, budgets $B_1, ..., B_c$, generator $\mathcal{G}$}
\KwResult{JLD sets $\Psi_1, \ldots, \Psi_c$ or \textsc{Infeasible}}

\For{each system retry $s = 1, \ldots, S$}{
    \For{each chain $\chain_i$}{
        \If{$\widehat{\emph{age}}(\chain_i, \Gamma, \emptyset) \leq B_i$}{
            $\Psi_i \gets \emptyset$;\ $\mathit{accepted}_i \gets \textsc{True}$\;
            \textbf{continue}; \tcp{No JLDs needed}
        }
        $\Psi_i \gets \emptyset$;\ $\mathit{accepted}_i \gets \textsc{False}$;\ $\Psi_i^{\mathrm{best}} \gets \emptyset$\;
        \For{each attempt $k = 1, \ldots, K$}{
            $\Psi \gets \mathcal{G}(\chain_i, \Gamma, k)$\;
            $\mathit{da}_k \gets \widehat{\emph{age}}(\chain_i, \Gamma, \Psi) \leq B_i$\;
            $\mathit{edf}_k \gets \textsc{EdfFeasible}(\chain_i, \Gamma, \Psi)$\;
            \If{$\mathit{da}_k$ \textbf{and} $\mathit{edf}_k$}{
                $\Psi_i \gets \Psi$;\ $\mathit{accepted}_i \gets \textsc{True}$\;
                \textbf{break}; \tcp{Candidate accepted}
            }
            \tcp{retain best-effort candidate}
            $\Psi_i^{\mathrm{best}} \gets \textsc{BestOf}(\Psi, \mathit{da}_k, \mathit{edf}_k,\; \Psi_i^{\mathrm{best}})$\;
            
        }
        \If{$\neg\,\mathit{accepted}_i$}{
            $\Psi_i \gets \Psi_i^{\mathrm{best}}$; \tcp{Carry best-effort into union}
        }
    }
    \tcp{Post-combine recheck against union $\Psi_1 \cup \cdots \cup \Psi_c$}
    \For{each chain $\chain_i$}{
        $\mathit{accepted}_i \gets \widehat{\emph{age}}(\chain_i, \Gamma, \bigcup_j \Psi_j) \leq B_i \;\textbf{and}\; \textsc{EdfFeasible}(\chain_i, \Gamma, \Psi_i)$\;
    }
    \If{$\bigwedge_i \mathit{accepted}_i$ \textbf{and} $\textsc{SystemEdf}(\Gamma,\, \Psi_1 \cup \cdots \cup \Psi_c)$}{
        \Return $\Psi_1, \ldots, \Psi_c$; \tcp{All checks passed}
    }
}
\Return \textsc{Infeasible}\;
\end{algorithm}

\subsection{DP Data-Age Checker}
Greedy-JLD~\cite{becker2016synthesizing} is exponential in the chain length $m$. Checking the worst-case data age for a given JLD set can be done in polynomial time with respect to the hyperperiod length. We argue that this complexity is acceptable for three reasons: (i) in real automotive systems, task periods are predominantly harmonic and hyperperiod explosion rarely occurs in practice (see~\cite{kramer2015waters, FinziRTAS24}), keeping $N = \hp/T_{\min}$ bounded and the test tractable, (ii) cause-effect chains in such systems can be very long, and (iii) every method for checking the feasibility of task sets with JLDs under EDF, including the Greedy-JLD heuristic itself, inherits the same complexity cost from the EDF demand-bound test~\cite{baruah1990algorithms, pellizzoni2005feasibility}. Hence, instead of paying both an exponential cost in chain length (during JLD synthesis) and an exponential cost in hyperperiod length (during EDF feasibility), we reduce the synthesis cost to the latter. Moreover, even when periods are not harmonic and $N$ grows large, several methods exist to transform the periods to harmonic variants~\cite{6932603, 10.5555/3288651.3288734, 7176034}.

\begin{algorithm}
\caption{DP Data-Age Checker}
\label{alg:dp-V2}
\KwData{chain $\chain = (\task_1, \ldots, \task_m)$, set of JLDs}
\KwResult{worst-case data age $l_{\max}$}

Compute $\hp = \mathrm{lcm}(T_1,\ldots,T_m)$ and $\WCL = \sum_{i=1}^{m}2T_i$\;
Set $H = \bigl(\lceil \WCL/\hp\rceil + 1\bigr)\cdot\hp$ and generate $\bar{N}_i=\max\{2,H/T_i\}$ jobs per task\;
Compute $[\rmin, \rmax]$, $[\dmin, \dmax]$ for each job; \tcp{initialise: one state per root job}
\For{each root job $j \in \{1, \ldots, \hp/T_1\}$}{
    $\mathcal{D}_{1,j} \gets \{\dmin(\tau_{1,j})\}$; $\mathcal{S}_{1,j,\dmin(j)} \gets \rmin(j)$\;
} 
\For{$\ell = 1, \ldots, m - 1$}{
    \For{each root job $j \in \{1, \ldots, \hp/T_\ell\}$}{
        \For{each $d' \in \mathcal{D}_{\ell, j}$}{
            \For{each each job of task $\tau_{\ell+1, s}$}{
                \If{$\rmax(\tau_{\ell+1, s}) \ge d'$ \textbf{and}   $\rmin(\tau_{\ell+1, s}) < \dmax(\tau_{\ell, j})$   
                    \textbf{and} $\neg\textsc{Blocked}(\tau_{\ell, j}, \tau_{\ell+1, s})$}{
                    \tcp{cascade $\dpmin$ (unified rule)}
                    $d'_{\mathrm{next}} \gets \max\!\bigl(d',\;\rmin(\tau_{\ell+1, s})\bigr) + \wcet_{\tau_{\ell+1, s}}$\;
                    \If{$\mathcal{D}_{\ell+1, s}$ is undefined} {
                        $\mathcal{D}_{\ell+1, s} \gets \{d'_{next}\}$;
                    } \Else {
                        $\mathcal{D}_{\ell+1, s} \gets \mathcal{D}_{\ell+1, s} \cup \{d'_{\mathrm{next}}\}$;\\
                    }
                    \If{$\mathcal{S}_{\ell+1, s, d'_{next}}$ is undefined}{
                        $\mathcal{S}_{\ell+1, s, d'_{next}} \gets S_{\ell, j, d'}$;
                    } \Else{
                        $\mathcal{S}_{\ell+1, s, d'_{next}} \gets \min\{\mathcal{S}_{\ell+1, s, d'_{next}}, S_{\ell, j, d'}\}$              
                    }
                }
                
            }
        }
    }
    \textsc{DominancePrune}$(\ell+1)$\;
}

\Return $\max_{s=1,2\ldots,HP/T_m}\max_{d \in \mathcal{D}_{m, s}} (\rmax(\tau_{m,s})+C_{m,s} -S_{m,s,d})$;\tcp{wc age at sink}
\end{algorithm}

We can compute $[\rmin, \rmax]$, $[\dmin, \dmax]$ for each job up to the interest of interval. For brevity, we denote them as $\rmin(\tau_{i,j}), \rmax(\tau_{i,j}), \dmin(\tau_{i,j}), \dmax(\tau_{i,j})$ by also dropping the virtual task notion. We note that the data propagation model assumes that each job is enforced delayed output or fixed execution time (instead of WCET) of every periodic task, detailed in Section~\ref{sec:task-model-chain-model}.

Let $\mathcal{D}_{\ell,j}$ be the set of earliest output times of job $\tau_{\ell,j}$. In our initial condition, we have $\mathcal{D}_{1,j} = \{\dmin(\tau_{1,j})\}$ for all $j=1,2,\ldots,\bar{N}_1$. We know that if a job $\tau_{\ell,j}$ can produce its output the earliest at time $d_{\ell,j}$, then a subsequent job of $\tau_{\ell+1,s}$ in the chain after reading this input can produce its output earliest at time $\max\{\rmin(\tau_{\ell+1,s}), d_{\ell,j}\}+C_{\ell+1,s}$. Formally, if a job $\tau_{\ell,j}$ produces its output earliest at time $d' \in \mathcal{D}_{\ell, j}$, then the data is read by another job $\tau_{\ell+1,s}$ earlist at time $\max\{d', \rmin(\tau_{\ell+1,s})$  if 1) $\rmax(\tau_{\ell+1, s}) \ge d'$, and 2) $\rmin(\tau_{\ell+1, s}) < \dmax(\tau_{\ell, j})$, and 3) $\tau_{\ell+1, s}$ does not block $\tau_{\ell, j}$. Furthermore, in this case, job $\tau_{\ell+1,s}$ can produce its output at earlist with $\max\{\rmin(\tau_{\ell+1,s}), d'\}+C_{\ell+1,s}$ in our data propagation model. 

Therefore, by using $\mathcal{D}_{\ell,j}$ for $j=1,2,\ldots,\bar{N}_\ell$ and the above conditions, we can build $\mathcal{D}_{\ell+1,s}$ for $s=1,2,\ldots,\bar{N}_{\ell+1}$.

Let $S_{\ell+1,s,d}$ be the earliest read time of the corresponding job of task $\tau_1$, whose output can be propagated to the job $\tau_{\ell+1, s}$  under the constraint that $\tau_{\ell+1, s}$  only produces its output earliest at time $d$. According to the definition, for every job $\tau_{1,j}$ with $d=\dmin(\tau_{1,j})$, we can initialize 
$S_{1,j,\dmin(\tau_{1,j})} = \rmin(\tau_{1,j}).$ 
For a given tuple of $\ell+1, s, d$ with $d \in \mathcal{D}_{\ell+1, s}$, let $\mathcal{J}_{\ell+1, s, d}$ be the tuple of $(\tau_{\ell, j}, d')$, which were used to define $\mathcal{D}_{\ell+1, s}$. By definition, $\mathcal{J}_{\ell+1, s, d} \neq \emptyset$. We define 
\begin{equation}
    \label{eq:recursive-S}
S_{\ell+1,s,d} = \min_{(\tau_{\ell,j}, d') \in \mathcal{J}_{\ell+1, s, d}}\{S_{\ell, j, d'}\}.  \end{equation}

\begin{lemma}
    \label{lemma:recursive}
    Suppose that task $\tau_{\ell+1, s}$ produces output the earliest at time $d \in \mathcal{D}_{\ell+1,s}$. 
    The longest valid data path from $\tau_1$ to $\tau_{\ell+1, s}$ producing the output at time $d$ is at most $d-S_{\ell+1,s,d}$.
\end{lemma}
\begin{proof}
    We prove this by induction. For $\ell=1$, this holds by construction. Suppose that this condition holds for the task $\tau_{\ell,j}$ for $j=1,2,\ldots,\bar{N}_{\ell}$. We now prove by the induction hypothesis that this also holds for $\tau_{\ell+1,s}$ for $s=1,2,\ldots,\bar{N}_{\ell+1}$.

    Suppose that there is a feasible data propagation path from a job $\tau_{1,k}$ to $\tau_{\ell+1, s}$ in which the job $\tau_{\ell+1, s}$ reads its input \emph{exactly} at time $d^* = d-C_{\tau_{\ell+1}, s}$. According to the definition of $\mathcal{D}_{\ell, j}$, we know that $d^* \in \mathcal{D}_{\ell, j}$ for some $j=1,2,\ldots,\bar{N}_{\ell}$. Suppose for contradiction that the time $\tau_{1,k}$ reads its input is \emph{less than} $S_{\ell+1,s,d}$.  Let $\tau_{\ell,j}$ be the job of task $\tau_{\ell}$ which produces its output the earliest at time $d^*$. Based on Eq.~\eqref{eq:recursive-S}, this contradicts the assumption that the longest valid data path from $\tau_1$ to $\tau_{\ell, j}$ producing the output at time $d^*$ is at most $d^* - S_{\ell,j,d^*}$.
\end{proof}

The above lemma assumes that $d\in \mathcal{D}_{\ell,s}$. We now show that such points are indeed sufficient. 

\begin{lemma}
    \label{lemma:dominance-1}
    Suppose that $\tau_{\ell+1, s}$ produces its output at time $d'$. Then a valid data propagation path from $\tau_1$ to $\tau_{\ell+1, s}$ is at most $\max_{d \in \mathcal{D}_{\ell+1,s}: d \leq d'} d' -S_{\ell+1,s,d}$.
\end{lemma}
\begin{proof}
    Suppose that this valid data propagation path $\tau_{\ell+1, s}$ reads its input at time $r$ and can produce its output the earliest at time $d = r+C_{\ell,s}$ with $d \leq d'$. Furthermore, $d \in \mathcal{D}_{\ell+1,s}$. According to Lemma~\ref{lemma:recursive}, we know that this data propagation path has a length of at most $d'- S_{\ell+1,s,d}$, as concluded in the statement of the lemma.
\end{proof}

\begin{theorem}
    \label{thm:dp-optimality}
   The longest valid data path from $\tau_1$ to $\tau_m$ is at most $\max_{s=1,2\ldots,HP/T_m}\max_{d \in D_{m, s}} (\rmax(\tau_{m,s})+C_{m,s} -S_{m,s,d})$.
\end{theorem}
\begin{proof}
    This comes directly from Lemmas~\ref{lemma:recursive}~and~\ref{lemma:dominance-1} and the fact that the output time of $\tau_{m,s}$ is at most $\rmax(\tau_{m,s})+C_{m,s}$.
\end{proof}

The above analysis shows the correctness of the DP approach. We can further prune $\mathcal{D}_{\ell,j}$ as follows. If $d' < d''$ are both in $\mathcal{D}_{\ell,j}$ and $S_{\ell,j,d'} \leq S_{\ell,j,d''}$, then we can remove $S_{\ell,j,d''}$ and also remove $d''$ from $\mathcal{D}_\ell$ if there is no further job defining that earliest output time. This pruning improves the complexity without sacrificing the optimality because: a smaller $d' < d''$ with an earlier starting time of $S_{\ell,j,d'} \leq S_{\ell,j,d''}$ passes the reachability check for a superset of successors and a smaller $S_{\ell,j,d'}$ yields an equal or larger data age at any leaf.

The DP data-age checker is presented in Algorithm~\ref{alg:dp-V2}. It runs in time
  $
    \bigO\!\left(\sum_{\ell=1}^{m-1}
      \bar{N}_\ell \cdot F_\ell \cdot \bar{N}_{\ell+1}\right)
    \;\le\; \bigO\!\left(m \cdot F_{\max} \cdot \bar{N}^2\right),
  $
where $\bar{N}_i$ is the number of jobs of task $\task_i$ within the analysis window, $\bar{N} = \max_i \bar{N}_i$, and $F_\ell$ is the number of non-dominated states per job (i.e., $|\mathcal{D}_{\ell, s}| \leq F_{\ell}$) at level~$\ell$ after dominance pruning. 
We note that the analysis and the DP algorithm can be easily adapted to deal with BCET $\neq$ WCET.

\subsection{GreedyRand-JLD Generator}
\label{sec:greedy-jld+}
We introduce a randomized variant of Greedy-JLD that we call \emph{GreedyRand-JLD}. The deterministic Greedy-JLD (Algorithm~3 in~\cite{becker2016synthesizing}) makes two fixed choices at each iteration: (a) it identifies the first (minimum) invalid path in the data propagation tree using a greedy depth-first search with oldest-candidate-first ordering, and (b) it walks the invalid path from leaf to root and places the JLD at the first \emph{mixed branch}, that is, a node whose subtree contains both valid and invalid paths, such that a sibling child leads to valid paths while the path-following child leads to invalid paths~\cite{becker2016synthesizing}. Note that the MECHAniSer implementation~\cite{becker2016mechaniser} simplifies step (b) to a free-pair scan, placing the JLD at the deepest edge that does not already carry a JLD, without explicitly checking for valid/invalid subtrees. In combination with the one-JLD-per-pair restriction, this guarantees termination in at most $m{-}1$ iterations but limits the algorithm to coarse-grained pruning.

Our GreedyRand-JLD variant preserves the same iterative structure but replaces both decision points with randomized choices that enable finer-grained control:
\begin{enumerate}[leftmargin=*, labelsep=0.5em]
  \item \textbf{Path selection.} Instead of the greedy DFS that returns the first invalid path found, the randomized variant builds the propagation tree (up to $500{,}000$ nodes with truncated trees using partial paths), annotates every subtree as valid ($\text{age} \le B$) or invalid ($\text{age} > B$), and selects \emph{uniformly at random} among all invalid complete or truncated root-to-leaf paths.
  \item \textbf{JLD placement.} Instead of always selecting the deepest mixed branch, the randomized variant collects \emph{all} mixed-branch nodes along the selected path and picks one \emph{uniformly at random}. A node is a mixed branch on a given path if (i)~it has both valid and invalid subtrees, (ii)~a sibling child has a valid subtree, and (iii)~the path-following child has an invalid subtree. The JLD is placed at the edge from the mixed branch to the path-following child, redirecting data flow from the invalid subtree toward the valid sibling.
  \item \textbf{Multiple JLDs per pair.} The one-per-pair restriction is removed: the algorithm may add several JLDs on the same $(\task_i, \task_j)$ edge with different local indices, enabling finer-grained pruning.
\end{enumerate}

The randomization is essential for diversity, i.e., each seed explores a different region of the JLD design space, and a wrong random choice is rejected by the Generate-and-Verify loop. Note that, unlike the deterministic Greedy-JLD which is bounded by $m{-}1$ iterations (one JLD per pair), the randomized variant may run up to $\mathit{max\_iter}$ iterations, since multiple JLDs per pair are allowed. The loop terminates early when the budget is satisfied or no invalid path remains and iterations that would emit a duplicate or unenforceable JLD are skipped (the algorithm continues to sample a different invalid path) rather than treated as termination. In our implementation we set $\mathit{max\_iter} = 50000$. 

\section{ML-Based Approach}
\label{sec:gnn}
We introduce a novel JLD synthesis approach based on machine learning (ML). The main insight is that cause-effect chains exhibit recurrent structural patterns connecting period configurations and their resulting JLDs, which a learning-based approach can exploit to generate candidate JLD sets efficiently, even on previously unseen task configurations. Note that we do not trust the ML output to be correct by construction and every candidate is verified by our DP data-age checker and the per-chain EDF feasibility checker within the Generate-and-Verify loop of Section~\ref{sec:generate_and_verify}. 

\begin{figure}
    \centering
    \includegraphics[width=0.7\linewidth]{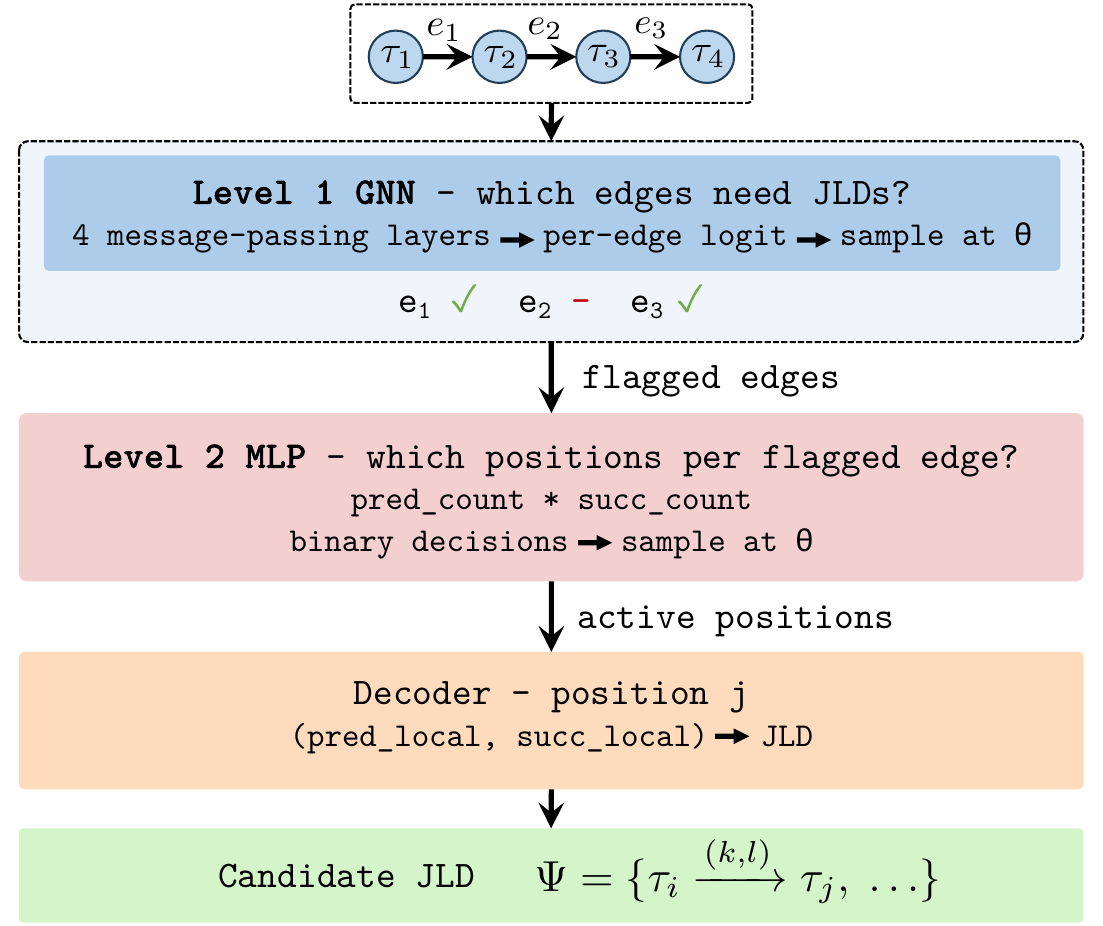}
    \caption{GNN inference pipeline. Level 1 (GNN) classifies each edge: JLD needed or not. For each flagged edge, Level 2 (MLP) classifies each (pred\_local, succ\_local) position. Both levels use temperature-controlled sampling $\theta$. The decoder maps active positions to a candidate $\Psi$ which is verified and integrated into the Generate-and-Verify framework (Algorithm~\ref{alg:gen-verify}).}
    \label{fig:inference}
\end{figure}

We adopt a two-level architecture (c.f., Figure~\ref{fig:inference}), motivated by the observation that the model needs to learn two structurally distinct decisions, namely \emph{which edges} in a chain need JLDs and, for those edges, \emph{which positions} to constrain.

\noindent \textbf{Level~1 (edge selection).} The first decision is which edges of the chain need JLDs at all. Note that the data-age budget is shared across the entire chain, hence constraining one edge affects how much slack remains for the others. We model this as a graph-level decision and employ a Graph Neural Network (GNN), which is naturally suited to the graph structure of cause-effect chains. The GNN learns to identify the bottleneck edges, i.e., those where the period mismatch contributes most to the worst-case data age. As an intuition, an edge between a $10$\,ms task and a $100$\,ms task admits $10$ candidate predecessor instances per pair-HP, whereas an edge between two $10$\,ms tasks admits essentially none. Thus, the GNN learns to focus JLDs where the propagation ambiguity is largest relative to the remaining budget.
The chain is modeled as a \emph{bidirectional graph}: $m$ nodes (tasks) connected by $m{-}1$ edges in both directions.
Each node carries 6 features: utilization $\wcet_i/T_i$; $\log_2(\hp/T_i)$ (jobs per HP); normalized chain position $i/(m{-}1)$; read-interval width ratio $(T_i - \wcet_i)/T_i$; budget ratio $B/\hp$; and deadline ratio $D_i/T_i$. Note that under the implicit-deadline assumption, the deadline ratio is identically~$1.0$ on every node and therefore contributes no learning signal.
Each edge carries 9 features: period ratio; log period disparity; harmonic indicator; budget tightness; normalized edge position; log chain length; reachability fraction; max and average branching factor.
A linear encoder maps node features to 64-dimensional hidden vectors and four message-passing layers with residual connections and layer normalization propagate context, each using a 2-layer MLP message function over concatenated source, target, and edge features. A 3-layer edge-classification head (hidden sizes $137{\to}64{\to}32{\to}2$, reading the two endpoint embeddings concatenated with the 9-D edge features) emits a binary JLD/no-JLD decision per edge. During training, an auxiliary head also predicts the chain's normalized max data age (MSE loss, weight~0.1), encouraging the shared backbone to learn features that reflect overall chain timing. This auxiliary head is discarded at inference and only the JLD predictions are used.

\noindent \textbf{Level~2 (position selection).} Once Level~1 has flagged the edges that need JLDs, Level~2 decides the specific JLD positions on each such edge. Note that this is a local decision, as it depends primarily on the period ratio of the two adjacent tasks rather than on the global structure of the chain. As an intuition, on a $10$\,ms-to-$100$\,ms (fast-to-slow) edge there are $10$ predecessor instances per pair hyperperiod, and the JLD typically constrains the successor to read from a recent predecessor instance, e.g., position $(1,1)$ encodes ``read from the first predecessor instance in each chain-HP group.'' Training labels generated with pair-HP indexing (e.g., by Greedy-JLD) are expanded to equivalent chain-HP positions, so that a same-rate pair-HP JLD $(1,1)$ becomes the diagonal $\{(1,1),(2,2),\ldots\}$ across the chain hyperperiod. Conversely, on a $100$\,ms-to-$10$\,ms (slow-to-fast) edge, the successor activates $10$ times per predecessor activation, and the JLD selects which of these $10$ successor instances must wait for fresh predecessor data. Since the relevant pattern depends on the period-ratio geometry rather than on global chain context, a simple Multi-Layer Perceptron (MLP) is sufficient. For each flagged edge, we use a 3-layer MLP (hidden sizes $7{\to}32{\to}32{\to}1$) that classifies the $\mathit{pred\_count} \times \mathit{succ\_count}$ binary positions, where each position encodes a pair $\langle \mathit{pred\_local}, \mathit{succ\_local} \rangle$ within the chain hyperperiod. Each position carries 7 features: normalized pred and succ local indices, both adjacent-task utilizations, $\log_2$ of the total position count, per-predecessor slack $(1{-}u_\mathrm{pred})/\mathit{pred\_count}$, and the centered pred position $(\mathit{pred\_local}/\mathit{pred\_count}){-}0.5$. A deterministic decoder converts the activated positions into JLD objects.

\subsection{Temperature-Controlled Sampling}

At inference, the model does not simply output its most likely prediction. Instead, logits are scaled by a temperature parameter $\theta$ before sampling. Level~1 uses softmax sampling: $p_k \propto \exp(z_k / \theta)$, selecting one class per edge. Level~2 uses independent sigmoid--Bernoulli sampling: each position is activated with probability $\sigma(z / \theta)$, so high temperature can activate many positions simultaneously.
Low temperatures ($\theta \to 0$) make both levels near-deterministic; high temperatures ($\theta \gg 1$) flatten toward uniform (Level~1) or $p{\approx}0.5$ per position (Level~2).
We use a fixed temperature ladder of 10 values spanning three orders of magnitude: $\theta \in \{0.01,\; 0.1,\; 0.3,\; 0.5,\; 0.8,\; 1.0,\; 1.5,\; 2.0,\; 3.0,\; 5.0\}.$
For $K{=}10$ attempts, each temperature is used once. For $K{>}10$, the ladder is traversed in the first 10 attempts and remaining attempts resample at $\theta{=}5.0$, providing independent stochastic draws that retain the model's learned bias while maximizing diversity. 
On system retries (Algorithm~\ref{alg:gen-verify}), the minimum temperature is shifted upward ($\theta_{\min} \in \{0.01, 0.5, 1.5\}$ for retries $0, 1, 2$), filtering the ladder to skip the near-deterministic region and increase exploration.

\begin{landscape}
\begin{figure}[p]
    \centering
    \includegraphics[width=0.95\linewidth]{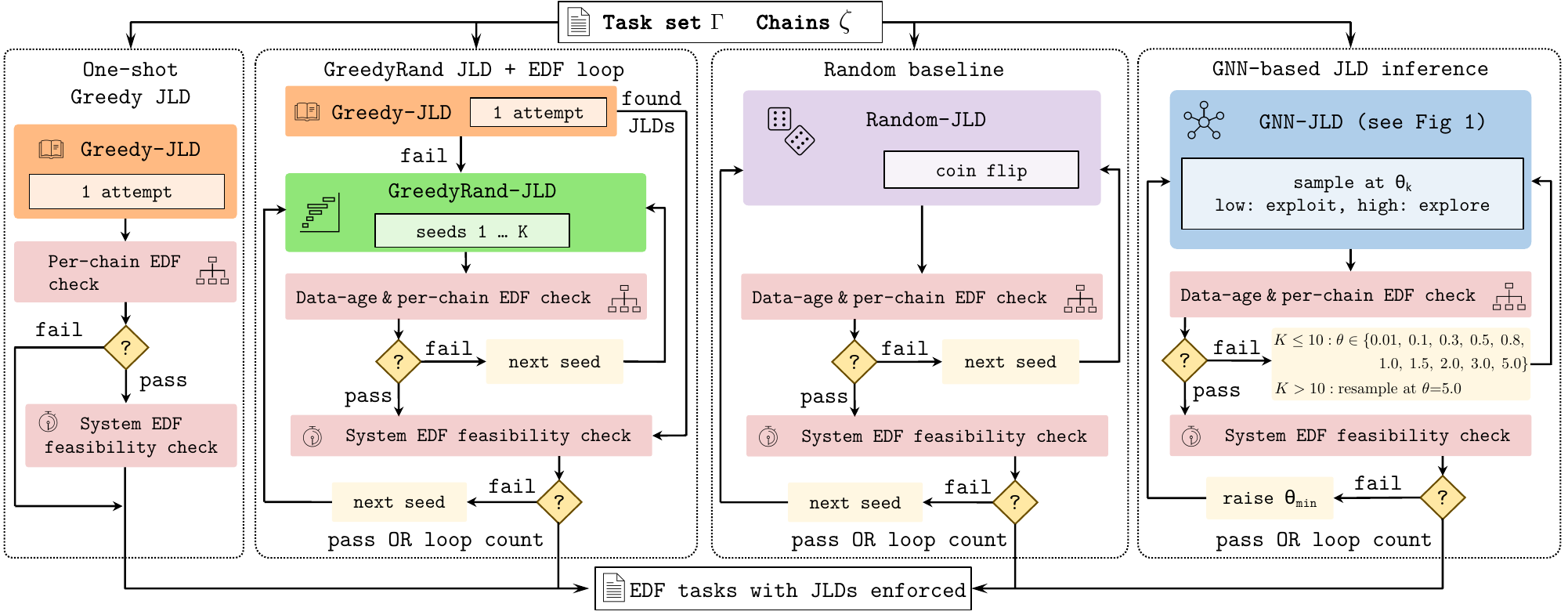}
    \caption{Benchmarking workflow. (a) Greedy-JLD: single deterministic attempt, followed by per-chain EDF enforceability and system EDF checks. (b) GreedyRand-JLD loop: Greedy-JLD first, then randomized heuristic seeds. (c) GNN-JLD loop: temperature-controlled sampling. (d) Random-JLD loop: uniform random baseline. All methods share the system-level EDF demand bound test; loop methods retry on system failure.}
    \label{fig:benchmark}
\end{figure}
\end{landscape}

\subsection{Training}
\sloppypar
Training data was generated synthetically with $4$ period sets: Set~A (automotive: $\{5, 10, 20, 40, 80\}$\,ms), Set~B ($\{8, 16, 24, 48, 96\}$\,ms), Set~C ($\{10, 30, 60, 90, 120\}$\,ms), Set~D ($\{10, 20, 25, 50, 100\}$\,ms). Task sets have utilization levels $U \in \{0.2, 0.4, 0.6, 0.8\}$ using UUniFast~\cite{bini2005measuring}, with 8--50 tasks, 1--10 chains of length 3--15, and implicit deadlines ($D_i = T_i$). Chain budgets are set to a fraction (0.3--0.9) of the unconstrained data age. Chains are generated with mixed period orderings: approximately 50\% fast-to-slow (common in automotive sensor-to-actuator chains, e.g.,~\cite{GuenzelBecker25RTAS}), 25\% slow-to-fast, and 25\% randomly shuffled. Training records are drawn from a fallback pipeline: Becker et al~\cite{becker2016synthesizing} is tried first; if either per-chain DA, per-chain EDF, or system-level EDF rejects, we fall through to a max-age Becker variant and then up to 10 randomized seeds with our GreedyRand-JLD method. Only DA-satisfying, system-EDF-feasible records are kept. Of the $12598$ records, 77.1\% come from~\cite{becker2016synthesizing}, 7.8\% from the variant, and 15.1\% from randomized seeds. Training data is dominated by fast-to-slow orderings (${\sim}66\%$), with slow-to-fast (${\sim}20\%$) and shuffled (${\sim}14\%$) minorities. The resulting training data contains chain-level records across all four period sets, with chain lengths $3$--$15$ and $1$--$27$ JLDs per chain. We note that the training data does not have to contain optimal JLDs but it is sufficient for them to be correct. Training on optimal JLD data is left for future work. Level~1 is trained for $1{,}000$ epochs with cross-entropy loss plus MSE loss for graph-level data age regression (weight $0.1$), using the Adam optimizer. Level~2 is trained for $2{,}000$ epochs with binary cross-entropy per position; the per-edge label space is $\mathit{pred\_count} \times \mathit{succ\_count}$.

\subsection{Baseline Methods}
\label{sec:baselines}

\begin{table}[t]
\centering
\footnotesize
\setlength{\tabcolsep}{2.7pt}
\renewcommand{\arraystretch}{1.15}
\caption{Simplified per-candidate complexity.
$m$ = chain length, $N = \hp/T_{\min}$, $n$ = total tasks in the system.}
\label{tab:complexity}
\begin{tabular}{@{}lccc@{}}
\toprule
 & \shortstack{Greedy-JLD /\\ GreedyRand-JLD}
 & GNN-JLD
 & Random-JLD \\
\midrule
\shortstack[l]{Generation (per cand.)}
 & $O\!\left(N \prod_i b_i\right)$
 & $O(mN^{2})$
 & $O(mN^{2})$ \\[2pt]

\shortstack[l]{DA + per-chain EDF}
 & internal / $O(mN^{2})$
 & $O(mN^{2})$
 & $O(mN^{2})$ \\[2pt]

System EDF
 & $O(n^3N^2)$
 & $O(n^3N^2)$
 & $O(n^3N^2)$ \\
\bottomrule
\end{tabular}
\end{table}

We compare $4$ JLD synthesis methods (Figure~\ref{fig:benchmark}). For generate-and-verify methods each candidate is checked by (i)~the DP data age checker and (ii)~the per-chain EDF feasibility checker (cf. Algorithm~\ref{alg:gen-verify}). A candidate is accepted only if both pass, otherwise the next attempt is tried. After all chains pass individually, a system-level EDF demand bound test checks combined schedulability, with up to $S$ system retries (cf. Algorithm~\ref{alg:gen-verify}). The $4$ JLD synthesis methods compared are: 
\begin{enumerate}[leftmargin=*, labelsep=0.5em]
    \item \textbf{Greedy-JLD} from~\cite{becker2016synthesizing, becker2016mechaniser}. Since~\cite{becker2016synthesizing, becker2016mechaniser} is deterministic, we do not use retries, but still do per-chain EDF feasibility checks and the final system-level EDF demand-bound test when reporting Enf.\ and Valid. 
    \item \textbf{GreedyRand-JLD}: out of the $K$ per-chain attempts, the first attempt runs deterministic Greedy-JLD (cf.\ Section~\ref{sec:greedy-jld+}). If that attempt fails the DA + per-chain EDF gate, the remaining $K-1$ attempts run the GreedyRand-JLD variant with distinct seeds.
    \item \textbf{Random-JLD baseline}: a pure-random baseline that has a generator based on uniform random sampling, i.e., each edge is independently flagged (probability 0.5), and each position is independently activated (probability 0.3). To keep every attempt non-trivial, if all edge flags miss we force one randomly chosen edge to be flagged, and if a flagged edge has no active position we force one randomly chosen position to be active. This baseline isolates the value of the GNN's learned structure: if the GNN performs similarly to random proposals with the same verification loop, it has not learned useful patterns.
    \item \textbf{GNN-JLD}: generates candidates via temperature-controlled sampling. Each of $K$ attempts samples from the logits at increasing temperature $\theta$, producing diverse JLD sets
    (see Section~\ref{sec:gnn}).
\end{enumerate}

We note that the four methods differ also in their JLD encoding. Greedy-JLD encodes each JLD at the pair hyperperiod $\mathrm{lcm}(T_{\mathrm{pred}}, T_{\mathrm{succ}})$, with at most one JLD per (pred-task, succ-task) pair (matching~\cite{becker2016mechaniser, becker2016synthesizing} encoding). GreedyRand-JLD, GNN-JLD, and Random-JLD methods encode each JLD at the chain hyperperiod $\hp = \mathrm{lcm}(T_i \mid \tau_i \in \chain)$ and allow multiple JLDs per pair. Each JLD carries its own reference HP and the verifier interprets each one accordingly, so the encodings co-exist in a single JLD set without further translation.

All three Generate-and-Verify methods (GreedyRand-JLD, GNN-JLD, Random-JLD) follow the same best-effort accounting: when a chain's $K$ attempts all fail, the loop retains the candidate that came closest to passing (per the ordering in Section~\ref{sec:generate_and_verify}) and contributes its JLDs to the per-method system union. The chain itself is still marked unsatisfied in isolation, but its best-effort JLDs participate in the post-combine DA recheck and the system-level EDF demand-bound test. Hence cross-chain shared-task tightening can rescue a chain whose own candidates failed in isolation. The single-attempt Greedy-JLD baseline does not have this loop structure, but it has an analogous property in MECHAniSer's incremental synthesis: JLDs are carried forward unconditionally across chains via cumulative state~\cite{becker2016synthesizing}, so the system union accumulates partial work from chains the synthesizer could not satisfy.

\subsection{Complexity}

\begin{landscape}
\begin{table}[p]
\caption{Benchmark results.  DA = DA-satisfied chains, Enf.\ = EDF-enforceable chains, Valid = fully valid use-cases (all chains DA + EDF + system EDF).  Med.\ = median per-chain generation time (ms).  JLD = mean JLDs per DA-satisfied chain. The bottom block reports ablation configurations with $K{=}1$ (no search budget).}
\label{tab:results}
\footnotesize
\begin{tabular}{l@{\;}rrrrr@{\quad}rrrrr@{\quad}rrrrr}
\toprule
& \multicolumn{5}{c}{\textbf{Known} (500 use-cases, 1257 chains)} & \multicolumn{5}{c}{\textbf{Similar} (500 use-cases, 1411 chains)} & \multicolumn{5}{c}{\textbf{Unknown} (500 use-cases, 1398 chains)} \\
\cmidrule(lr){2-6} \cmidrule(lr){7-11} \cmidrule(lr){12-16}
Method & DA & Enf & Val & Med & JLD & DA & Enf & Val & Med & JLD & DA & Enf & Val & Med & JLD \\
\midrule
Random-JLD ($K{=}10$)
    & 815 & 383 & 79 & 72 & 16.1
    & 805 & 384 & 51 & 86 & 17.8
    & 863 & 411 & 55 & 47 & 19.6 \\
Greedy-JLD
    & 1146 & 1121 & 384 & 27 & \best{3.8}
    & 1176 & 1147 & 282 & 30 & \best{4.1}
    & 1126 & 1113 & 298 & 23 & \best{3.6} \\
GreedyRand-JLD ($K{=}10$)~~~~~~~
    & 1211 & 1211 & 455 & 31 & \second{4.2}
    & 1300 & \second{1296} & 386 & 41 & \second{5.1}
    & 1227 & 1227 & 358 & 34 & \second{4.5} \\
GNN-JLD ($K{=}10$)
    & \second{1240} & \second{1220} & \second{468} & \best{6} & 9.7
    & \second{1330} & 1263 & \second{392} & \best{8} & 13.8
    & \second{1273} & \second{1239} & \second{391} & \best{7} & 11.4 \\
GNN-JLD ($K{=}50$)
    & \best{1257} & \best{1248} & \best{487} & \best{6} & 9.7
    & \best{1389} & \best{1330} & \best{434} & \best{8} & 14.5
    & \best{1335} & \best{1298} & \best{416} & \best{7} & 11.8 \\
\midrule
\multicolumn{16}{l}{\emph{Ablation (no search budget):}} \\
Random-JLD ($K{=}1$)
    &  235 &   95 &   6 &  8 & 11.3
    &  232 &  113 &   2 &  8 &  9.9
    &  244 &   94 &   3 &  6 & 11.6 \\
GNN-JLD ($K{=}1$, argmax)
    &  904 &  904 & 257 &  4 &  7.9
    &  814 &  814 & 151 &  5 & 11.8
    &  820 &  809 & 160 &  4 &  8.8 \\
\bottomrule
\end{tabular}
\end{table}
\end{landscape}

Table~\ref{tab:complexity} compares per-candidate costs in terms of chain length~$m$ and $N = \hp/T_{\min}$.
The heuristic methods are dominated by propagation-tree construction: $O(N \cdot \prod b_i)$ per candidate, where $b_i = \lceil(2T_i - \wcet_i) / T_{i+1}\rceil + 1$ is the branching factor at edge~$i$, giving exponential growth in~$m$. Each candidate of the GreedyRand-JLD internally adds JLDs one at a time, rebuilding the propagation tree at every iteration up to $\mathit{max\_iter}$ (duplicate detection limits most runs to fewer iterations). The GNN-JLD and Random-JLD methods generate and verify each candidate in $O(m N^2)$. 
The Generate-and-Verify loop runs up to $K$ candidates per chain; the total per-chain cost is therefore $O(K \cdot m \cdot N^2)$ for GNN-JLD/Random-JLD vs.\ $O(K \cdot N \cdot \prod b_i)$ for the heuristic. The system-level EDF demand bound test (Lemma~3.4 of~\cite{baruah1990algorithms}) runs once per use-case over the virtualized task set in $O(|\Lambda| \cdot |\Delta| \cdot n)$, where $|\Lambda|, |\Delta| \leq O(n \cdot N)$ are the test point set sizes; bounding yields $O(n^3 N^2)$. Table~\ref{tab:complexity} simplifies several constant factors for readability. The DP checker's frontier size $F_{\max}$ equals~$1$ for harmonic periods, reducing $O(m F_{\max} N^2)$ to $O(m N^2)$. The number of root jobs $R = \hp / T_1 \leq N$ is absorbed into the $N$ factor of the heuristic tree cost. GNN-JLD hidden dimensions ($h, h_2$) and per-chain EDF transform costs ($O(V{+}E)$) are omitted as they are dominated by the $O(m N^2)$ terms.

\begin{figure*}[!ht]
\centering
\includegraphics[width=0.99\linewidth]{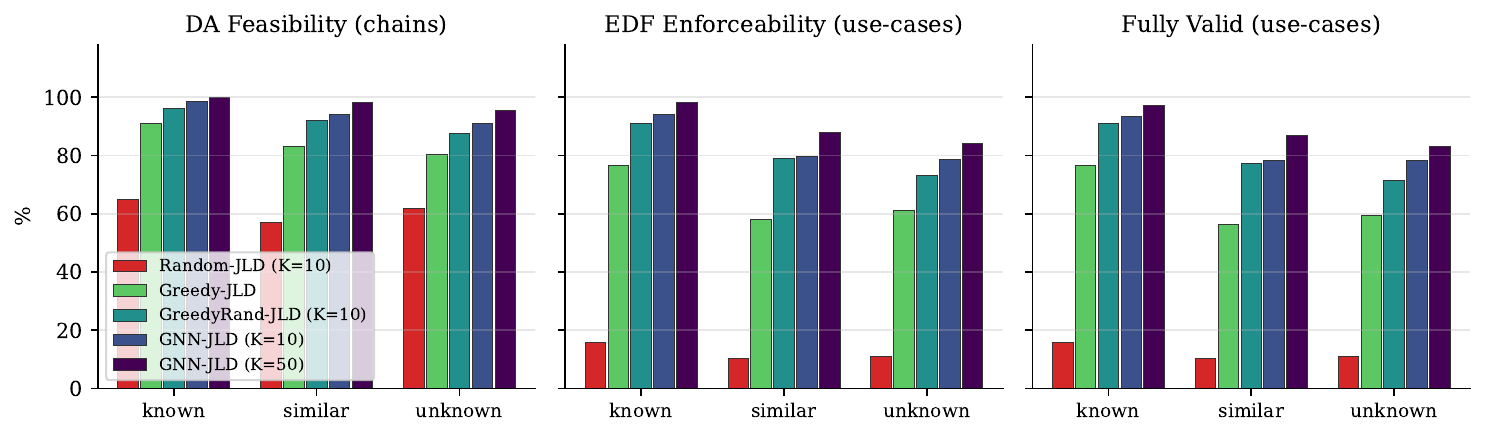}
\caption{Success rates across all benchmark sets.
Left: DA feasibility per chain.
Center: EDF enforceability per use-case (all chains individually DA-feasible and EDF-enforceable).
Right: fully valid per use-case (all chains DA + EDF enforceable + system-level EDF schedulable).}
\label{fig:success_rates}
\end{figure*}

\section{Evaluation}
\label{sec:evaluation}

\subsection{Experimental Setup}

The benchmark consists of three test sets, each containing $500$ task sets: \emph{Known}, randomly chosen from the training data; \emph{Similar}, drawn from the same period families and chain characteristics as the training data but with a new random seed; and \emph{Unknown}, with synthetic period sets never seen in training, whose values are drawn from $\{3,4,5,6,8,9,12,15,16,18,30,32,36,40,54,60,64,120\}$\,ms. Each task set in \emph{Unknown} contains $12$--$40$ tasks, $1$--$7$ chains of length $4$--$12$, and a system utilization $U \sim \mathcal{U}(0.15, 0.85)$.

\begin{figure}[t]
\centering
\includegraphics[width=0.75\columnwidth]{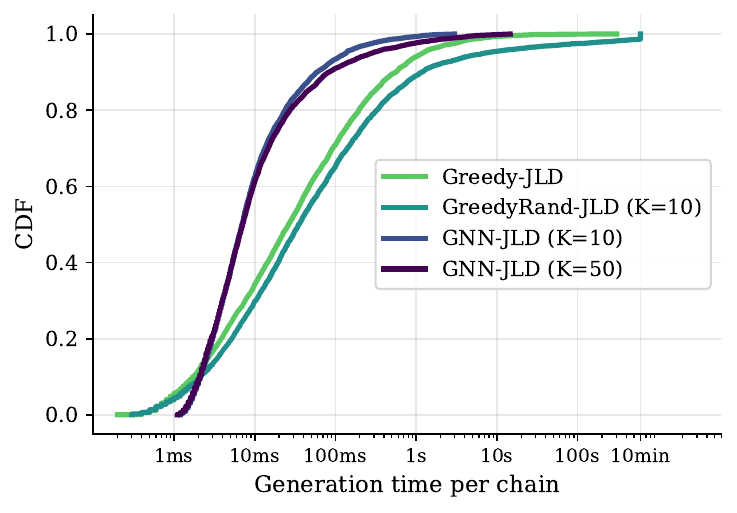}
\caption{CDF of per-chain generation time (log scale). GNN-JLD completes the majority of chains in under 50\,ms (median 6--8\,ms). The heuristic methods span three orders of magnitude.}
\label{fig:gen_time_cdf}
\end{figure}

For each method, we report four levels of success: (i) \emph{DA satisfied} (per chain), the data age is within budget; (ii) \emph{EDF enforceable} (per chain), the DA-satisfied JLDs pass the per-chain precedence transform with $D \geq \wcet$ after adjustment; (iii) \emph{System EDF} (per use-case), the combined demand-bound test passes, which we evaluate only when all chains in the use-case satisfy DA and per-chain EDF; and (iv) \emph{Fully valid} (per use-case), all three levels hold for all chains in the use-case. We additionally report the average number of JLDs per DA-satisfied chain, the median per-chain and per-use-case runtime, the average number of attempts for retry-capable methods, and the DA success rate broken down by chain length.

\subsection{Results}

All experiments were run on a single Apple M4 core (no GPU). For all loop methods we use $K{=}10$ per-chain attempts and $S{=}3$ system retries, with a per-chain timeout of $600$\,s and chains exceeding the timeout are reported as unsatisfied. We additionally evaluate a configuration of GNN-JLD with $K{=}50$ per-chain attempts to study the effect of an extended compute budget. Table~\ref{tab:results} summarizes the aggregate results, and the figures below break down individual metrics.

\textbf{Success rates.}
Figure~\ref{fig:success_rates} and Table~\ref{tab:results} present the multi-level success metrics (numbers in green show the best value per column and in orange show the second best per column). At the DA level, GNN-JLD($K{=}50$) achieves the highest chain-level satisfaction across all benchmark sets: 100.0\% on Known, 98.4\% on Similar, and 95.5\% on Unknown. The single-attempt Greedy-JLD achieves 91.2\%, 83.3\%, and 80.5\% respectively, while GreedyRand-JLD improves this to 96.3\%, 92.1\%, and 87.8\% through randomized retries. For fully-valid use-cases, GNN-JLD($K{=}50$) achieves the highest valid rate across all benchmark sets: 97.4\% on Known, 86.8\% on Similar, and 83.2\% on Unknown, compared to 91.0\%, 77.2\%, and 71.6\% for GreedyRand-JLD, demonstrating generalization to unseen period sets. The gap between DA satisfaction and full validity reflects the EDF enforceability bottleneck. GNN-JLD produces more JLDs per chain (mean 10--15 vs.\ 3.6--4.1 for Greedy-JLD and 4.2--5.1 for GreedyRand-JLD), and some configurations create scheduling constraints infeasible under EDF. The Random baseline validates the learned structure of GNN-JLD since with the same decoder and verification loop, random proposals achieve 10--16\% fully-valid use-cases compared to 78--94\% for GNN-JLD($K{=}10$).

\begin{figure}[!t]
\centering
\includegraphics[width=0.75\columnwidth]{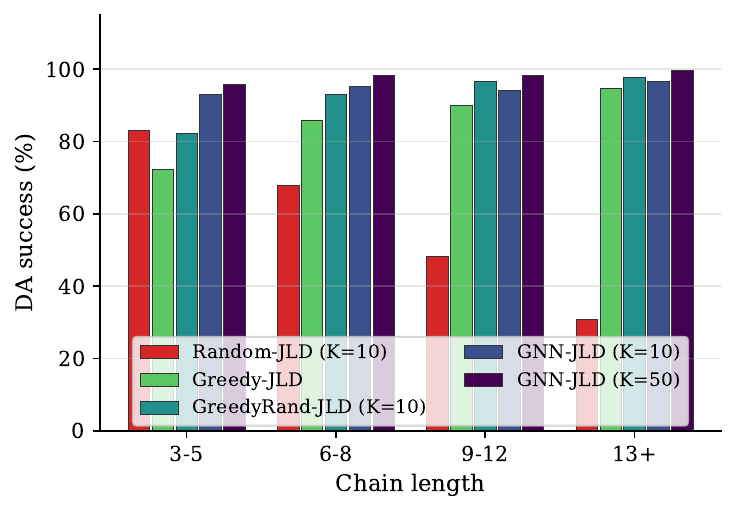}
\caption{DA satisfaction rate by chain length across all benchmark sets.}
\label{fig:da_by_len}

\end{figure}

\begin{figure}[!t]
\centering
\includegraphics[width=0.75\columnwidth]{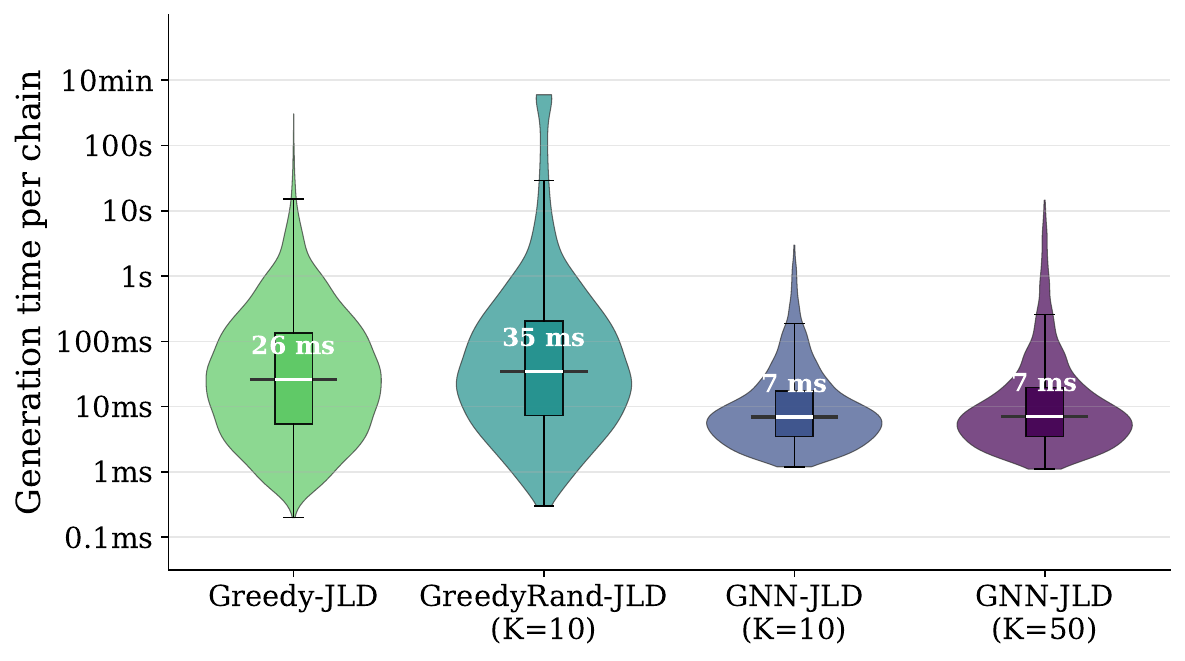}
\caption{Per-chain generation time (log scale).
The inner box is the interquartile range and the white line is the median.}
\label{fig:gen_time_box}

\end{figure}

\begin{figure}[!ht]
\centering
\includegraphics[width=0.75\columnwidth]{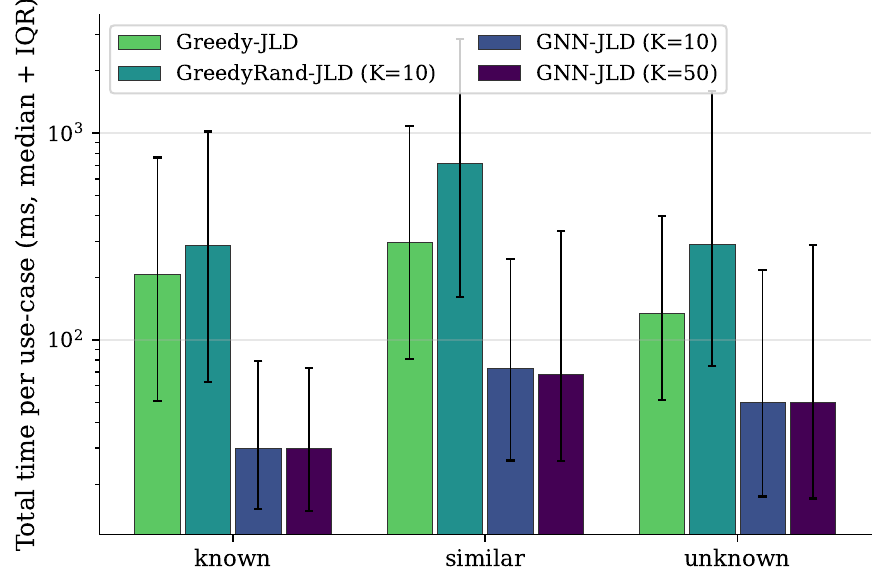}
\caption{Median total time per use-case (generation + system EDF).}
\label{fig:total_time_median}

\end{figure}

\begin{figure*}[!ht]
\centering
\includegraphics[width=0.99\linewidth]{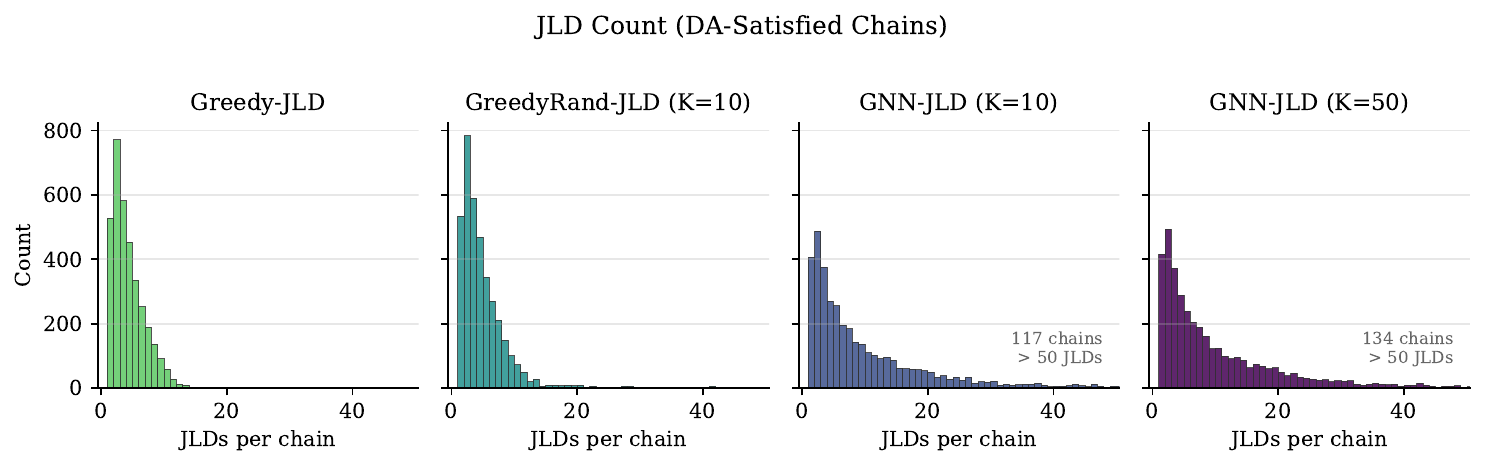}
\caption{Distribution of JLDs per DA-satisfied chain.
Greedy-JLD produces the fewest JLDs due to its one-per-pair rule.}
\label{fig:jld_count}

\end{figure*}

A central question is whether GNN-JLD's gain over Greedy-JLD stems from the $K$-attempt generate-and-verify loop and not from the learned structure. We address this with the two $K{=}1$ ablation rows in Table~\ref{tab:results}, where \emph{GNN-JLD ($K{=}1$, argmax)} runs the trained model with a single near-zero-temperature sample ($\theta < 10^{-6}$, no search), and \emph{Random-JLD ($K{=}1$)} runs a single uniform-random run. With $K{=}1$  GNN-JLD outperforms Random-JLD by roughly $3$--$4\times$ in DA satisfaction (71.9\% vs.\ 18.7\% on Known, 57.7\% vs.\ 16.4\% on Similar, 58.7\% vs.\ 17.5\% on Unknown), showing that there is learned structure independent of the search loop. With $K{=}10$, random search alone reaches only 15.8\,/\,10.2\,/\,11.0\% fully-valid use-cases, well below the deterministic Greedy-JLD baseline (76.8\,/\,56.4\,/\,59.6\%), so the $K$-attempt loop by itself does not explain GNN-JLD's advantage.

\begin{figure}[!ht]
\centering
\includegraphics[width=0.75\columnwidth]{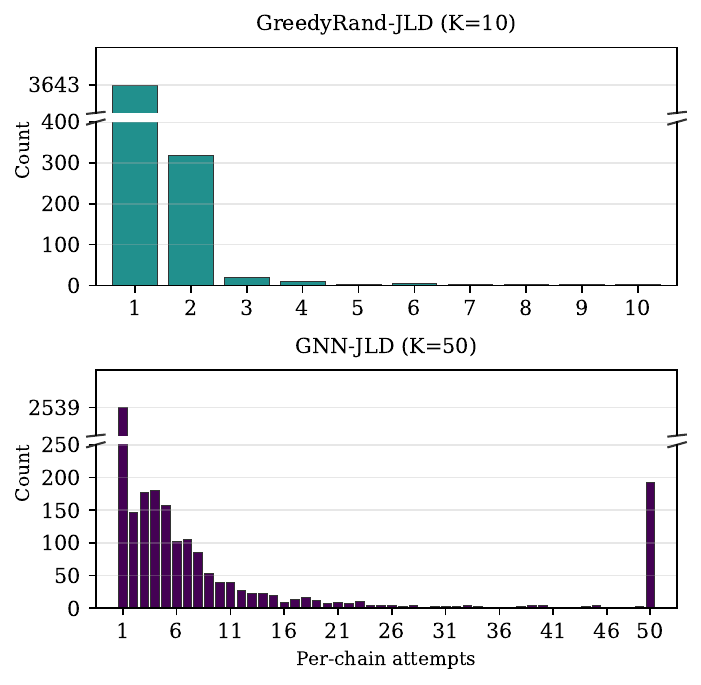}
\caption{Distribution of per-chain attempts for retry-capable methods.}
\label{fig:attempts}
\end{figure}

\textbf{Chain length sensitivity.}
Figure~\ref{fig:da_by_len} shows DA success by chain length. All methods maintain relatively stable rates across lengths 3--15. GNN($K{=}50$) achieves 92--100\% DA satisfaction across all length bins (lowest at Unknown $3$--$5$ with 91.9\%, reaching 100\% on every Known bin). GNN-JLD($K{=}10$) covers a wider range (89--99\%), with the lowest bin at Unknown $3$--$5$ (88.9\%) since chains in the Unknown set exhibit period combinations not seen during training. Increasing the attempts to $K{=}50$ recovers most of these cases: Unknown $3$--$5$ rises to 91.9\% and Similar $3$--$5$ rises from 93.7\% to 96.3\%. GreedyRand-JLD spans 80--100\% across length bins, with the highest rates at longer Known chains (Known $13{+}$ and $9$--$12$: 100\%) and the lowest at shorter Similar chains (Similar $3$--$5$: 79.7\%). 
Random-JLD performs reasonably on the 3–5-task chains because the JLD configuration space is small enough that $K{=}10$ uniform samples cover it densely and a single bottleneck edge often suffices. However, it collapses on longer chains, where the search space grows combinatorially.

\textbf{Generation time.}
Both GNN-JLD variants have a median of 6--8\,ms; GNN-JLD($K{=}10$) completes about 59--70\% of chains under 10\,ms and 86--94\% under 50\,ms, while GNN-JLD($K{=}50$) reaches 57--68\% under 10\,ms and 82--92\% under 50\,ms. The heuristic methods exhibit heavy-tailed distributions: Greedy-JLD has a median of 23--30\,ms with a tail extending beyond minutes, while GreedyRand-JLD reaches the 600\,s per-chain timeout on hard chains due to propagation tree explosion (18 timeouts on Similar, 41 on Unknown). Figure~\ref{fig:gen_time_box} shows the full distributions, confirming the GNN's tight concentration around 6--8\,ms against the heuristic spread across four orders of magnitude. At the median, GNN-JLD is \textbf{3--5$\times$ faster}; in the tail, the speedup exceeds \textbf{100$\times$}.

Figure~\ref{fig:total_time_median} shows the end-to-end runtime per use-case, covering all per-chain generation (including system retries) plus the system-level EDF demand bound test. The advantage of GNN-JLD is amplified at the use-case level: median ${\sim}$30--73\,ms vs.\ ${\sim}$133--294\,ms for single-attempt Greedy-JLD (\textbf{3--7$\times$ speedup}) and ${\sim}$285--718\,ms for GreedyRand-JLD (\textbf{6--10$\times$ speedup}).

\textbf{JLD count.}
Figure~\ref{fig:jld_count} shows the JLD count distribution for DA-satisfied chains. The heuristic methods produce concentrated counts (mean 3.6--5.1), structurally limited by the one-per-pair rule.
GNN-JLD produces wider distributions (mean 10--15), with a long tail: 117 chains for GNN-JLD($K{=}10$) and 134 for GNN-JLD($K{=}50$) exceed 50 JLDs. The random and ML-based methods use the chain hyperperiod (LCM of all task periods) for JLD local indexing, which yields finer-grained, per-edge constraints that repeat less frequently than Becker's pair hyperperiod (LCM of only the two adjacent task periods), requiring potentially more JLDs.

\textbf{Attempt efficiency.}
Figure~\ref{fig:attempts} shows the per-chain attempt distribution.
GreedyRand-JLD solves the majority of chains in the first attempt using the deterministic Becker attempt (attempt~1). GNN-JLD($K{=}50$) is front-loaded similarly, i.e., the first temperature ($\theta_0 = 0.01$, near-deterministic) succeeds for most chains, but the gradual tail through increasing temperatures recovers additional cases, with fewer chains exhausting all 50 attempts compared to $K{=}10$.

\textbf{Per-use-case comparison.}
Figure~\ref{fig:scatter} compares per-use-case normalized data age between Greedy-JLD and GNN-JLD($K{=}50$) on the Unknown set. Each axis shows the worst chain's ratio $\text{max\_data\_age} / B$. Values $\leq 1$ satisfy the data-age budget and values $> 1$ exceed it. The two dotted red lines at $x{=}5$ and $y{=}5$ mark the over-constrained penalty bin, when a method's JLDs block all propagation paths so that no data can flow from sensor to actuator. Points below the diagonal: GNN-JLD achieves lower (better) normalized age. Quadrants reflect data-age satisfaction only, i.e., marker color reflects full validity (DA + per-chain EDF + system EDF): green = both fully valid, dark blue = only GNN-JLD fully valid, red = only Greedy-JLD fully valid, gray = neither.
A color/quadrant mismatch, e.g. grey marker inside the blue quadrant, means e.g. that a use-case in which GNN-JLD finds JLDs with DA satisfied but the EDF check rejected them.

\begin{figure}[!t]
\centering
\includegraphics[width=0.8\columnwidth]{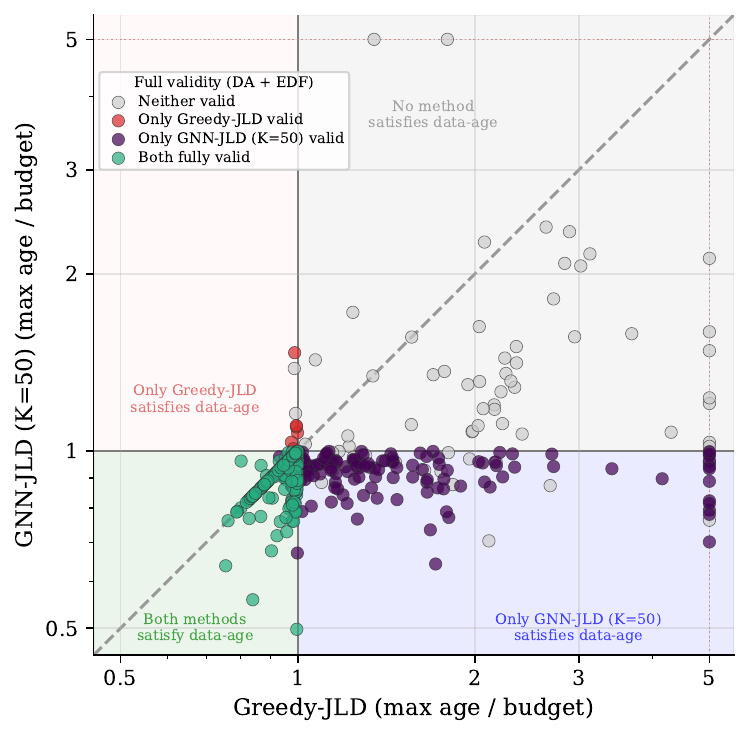}
\caption{Normalized data age per use-case for Greedy-JLD vs.\ GNN-JLD($K{=}50$) on the Unknown set.}
\label{fig:scatter}
\end{figure}

The ``Only GNN-JLD valid'' cluster in the blue quadrant represents the 130 use-cases solved exclusively by the GNN-JLD and only 12 use-cases are valid only under Greedy-JLD (but most of them are close to being fulfilled also with GNN-JLD). When both methods succeed (green markers in the green quadrant), GNN-JLD achieves comparable normalized age despite using chain-HP indexing with more JLDs. However, when both fail (grey quadrant), most of the use-cases have lower cost with GNN-JLD($K{=}50$). The vertical/horizontal clustering along these red lines is denser on the GNN-JLD $y$-axis because GNN-JLD produces more JLDs per chain, occasionally over-constraining chains that the Greedy heuristic leaves unconstrained (but also unsatisfied). A small number of markers also break the colour--quadrant correspondence, where both method data ages fit their budgets but the GNN-JLD candidate JLDs were rejected by an EDF check and these are exactly the use-cases in which the per-chain or system-level EDF feasibility test, rather than data age, decides the outcome.

\begin{figure}[!h]
\centering
\begin{subfigure}[t]{0.48\columnwidth}
    \centering
    \includegraphics[width=\linewidth]{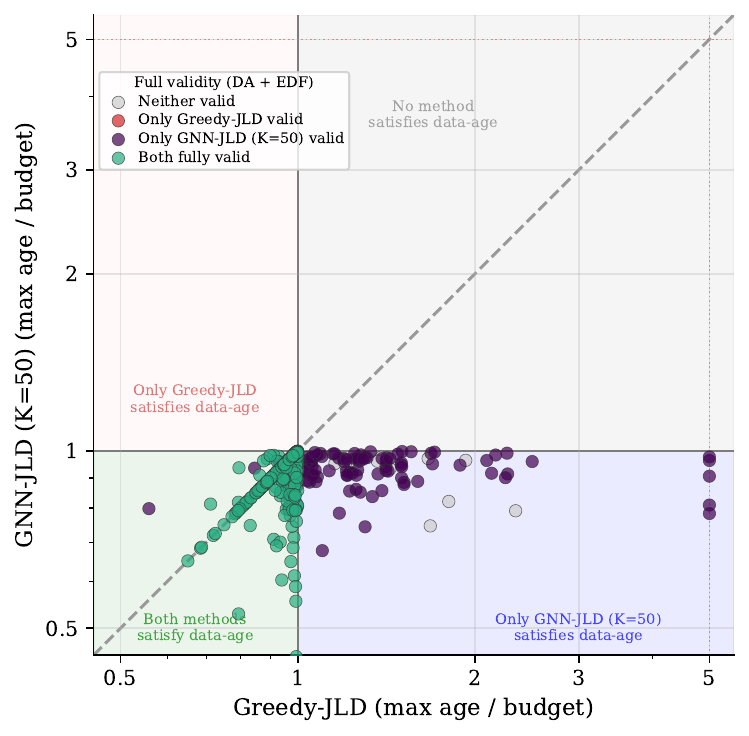}
    \caption{Known set.}
    \label{fig:scatter_known}
\end{subfigure}
\hfill
\begin{subfigure}[t]{0.48\columnwidth}
    \centering
    \includegraphics[width=\linewidth]{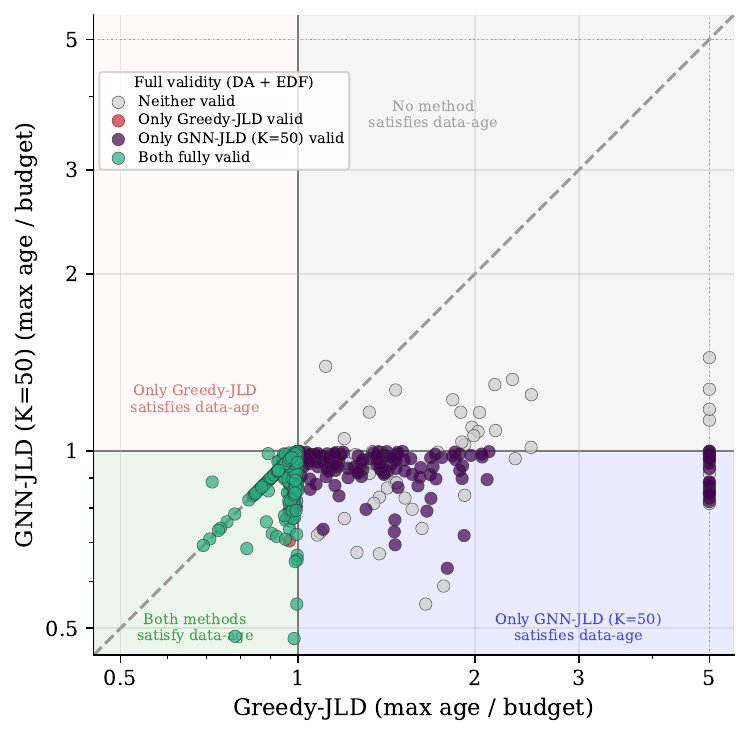}
    \caption{Similar set.}
    \label{fig:scatter_similar}
\end{subfigure}
\caption{Normalized data age per use-case for Greedy-JLD vs.\ GNN-JLD($K{=}50$).}
\label{fig:scatter_known_similar}
\end{figure}

Figures~\ref{fig:scatter_known} and \ref{fig:scatter_similar} show the same per-use-case normalized data age graph between between Greedy-JLD and GNN-JLD($K{=}50$) on the Known and Similar sets, respectively.

Figures~\ref{fig:scatter_br_unknown}, \ref{fig:scatter_br_known}, and \ref{fig:scatter_br_similar} show the per-use-case normalized data age graph between GreedyRand-JLD (K=10) and GNN-JLD($K{=}50$) on the Unknown, Known and Similar sets, respectively.

\begin{figure}[!h]
\centering
\begin{subfigure}[t]{0.48\columnwidth}
    \centering
    \includegraphics[width=\linewidth]{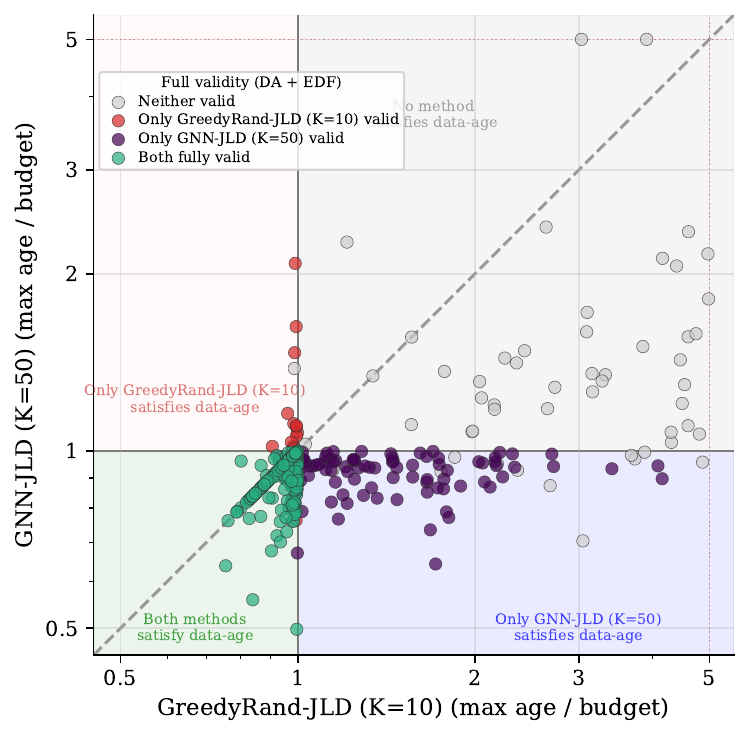}
    \caption{Unknown set.}
    \label{fig:scatter_br_unknown}
\end{subfigure}
\hfill
\begin{subfigure}[t]{0.48\columnwidth}
    \centering
    \includegraphics[width=\linewidth]{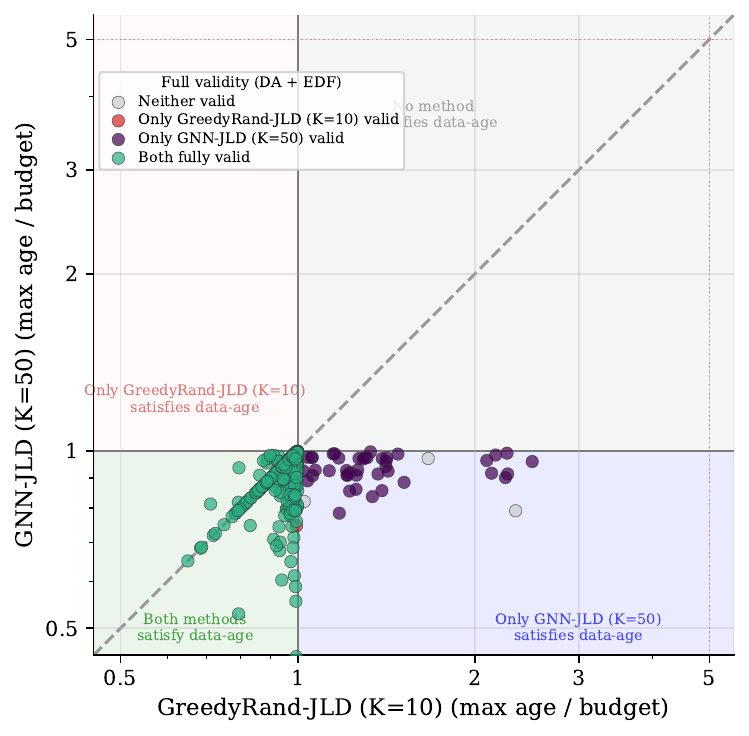}
    \caption{Known set.}
    \label{fig:scatter_br_known}
\end{subfigure}

\vspace{0.5em}

\begin{subfigure}[t]{0.48\columnwidth}
    \centering
    \includegraphics[width=\linewidth]{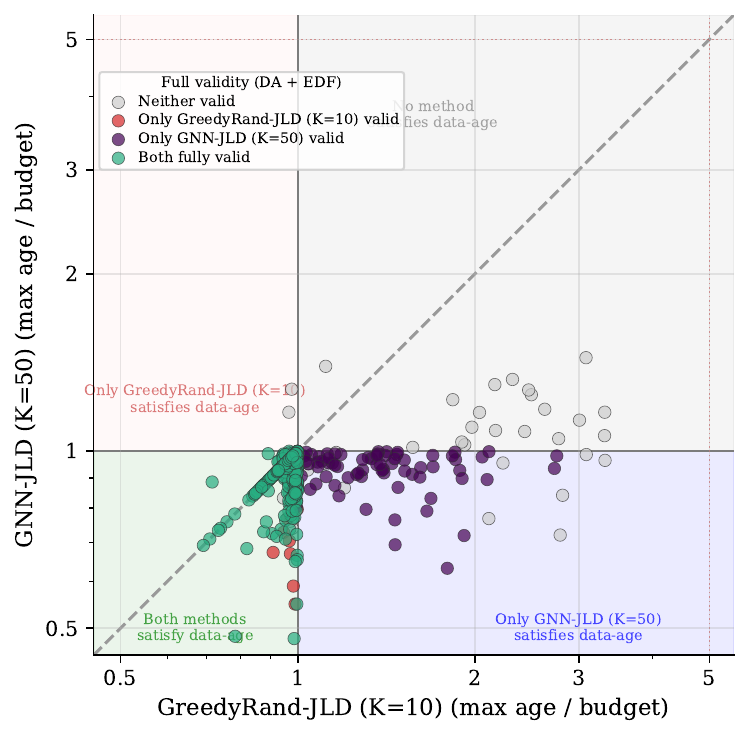}
    \caption{Similar set.}
    \label{fig:scatter_br_similar}
\end{subfigure}
\caption{Normalized data age per use-case for GreedyRand-JLD vs.\ GNN-JLD($K{=}50$).}
\label{fig:scatter_br_gnn}
\end{figure}

In all cases we see that the ``Only GNN-JLD valid'' cluster in the blue quadrant in which the cases solved exclusively by the GNN-JLD are shown are the majority of use-cases compared to the other method. Moreover, when both methods succeed (green markers in the green quadrant) or when both methods fail, GNN-JLD achieves either comparable or lower normalized age despite using chain-HP indexing with more JLDs. This is particularly visible when both fail (grey quadrant) since most of the use-cases have lower cost with GNN-JLD($K{=}50$). 

\section{Related Work}
\label{sec:related}
Cause-effect chains and their schedulability has been extensively studies with classical analysis methods~\cite{Guenzel-TECS24-end-to-end-tutorial}. Hamann et al.~\cite{Hamann-ECRTS17-communication-centric} present the Bosch communication-centric design that motivates much of this line of work, while Schlatow et al.~\cite{Schlatow-SIES18-data-age-analysis-optimisation} and Dürr et al.~\cite{Duerr-TECS19-sporadic-cause-effect-chains} extended the analysis to optimisation and to sporadic chains in distributed systems, respectively.  Subsequent work by Günzel et al.~\cite{GuenzelBecker25RTAS} optimises task phasing in harmonic and semi-harmonic settings.  JLDs were introduced by Becker et al.~\cite{becker2016synthesizing,becker2018endtoend, becker2016mechaniser} as a schedule-agnostic mechanism for constraining the worst-case data age of cause-effect chains together with the Greedy-JLD heuristic. However, the Greedy-JLD heuristic does not check whether the synthesised JLDs are enforceable under any concrete scheduler. Klaus et al.~\cite{klaus2019constrained} studied the run-time overheads of enforcing JLDs in real systems and analysed the trade-off between data-age tightness and synchronisation cost, but their work assumes the JLD set is already given. Our paper directly addresses the synthesis side of this gap by producing JLD sets that are simultaneously data-age feasible \emph{and} enforceable under EDF. Forget et al.~\cite{forget} studied FP scheduling under multi-rate depedencies through release-time/deadline adjustment and priority assignment. 
More recently, Becker~\cite{becker_merging} proposed eliminating runtime JLD enforcement by merging tasks connected by JLDs into a multi-frame task using a CP formulation, whereas our work synthesizes JLD sets and verifies their data-age correctness, EDF enforceability, and system-level schedulability.

Using Machine Learning (ML) in real-time systems has increased recently~\cite{Editorial-RTS25-learning-roadmap}. Several papers in~\cite{Editorial-RTS25-learning-roadmap} argue for integration of ML with scheduling theory: Guo et~al.~\cite{Guo-RTS25-ML-marry-RT-scheduling} survey the use of NN in real-time scheduling, while Casini~\cite{Casini-RTS25-MILP-or-AI} asks when AI-based heuristics outperform exact MILP formulations. A recurring theme is that learned policies, while fast, must be paired with verifiers that guarantee correctness. This is what we achieve with our Generate-and-Verify architecture instantiated for the JLD synthesis problem. Similarly, a polynomial-time verifier is employed in~\cite{10.1007/s11241-025-09450-y} for learning-assisted schedulability analysis. Heinz et al.~\cite{heinz2025rl} apply RL to constraint-programming search, demonstrating that learning-guided exact methods can match or exceed hand-tuned solvers. ML has been recently used for scheduling problems. In~\cite{Zhou-ICAICE23-multicore-GNN-DRL} a graph attention network is combined with deep reinforcement learning to schedule DAGs on multi-core processors. A GNN-for-scheduling survey for generic job shop and flow shop scheduling problems is presented in~\cite{SMIT2025106914} while~\cite{BOUSKA2023990} studies learning-driven scheduling for single-machine and multi-machine tardiness problems. No prior work applies ML to the synthesis of \emph{JLDs for cause-effect chains}. 

\section{Conclusion}
\label{sec:conclusion}

We introduced a \emph{Generate-and-Verify} framework for synthesizing enforceable and schedulable job-level dependencies in EDF-scheduled cause-effect chains, pairing fast candidate generators with a safe DP data-age checker and an EDF feasibility transform. Within this framework we proposed two novel generators: a randomized variant of the greedy heuristic that lifts the one-JLD-per-pair restriction, and an ML-based generator combining a GNN for edge selection with a per-edge MLP for position selection, sampled with a temperature-controlled schedule. Our evaluation shows that both generators outperform the original greedy heuristic in fully-valid use-cases (including on period configurations unseen during training) and that the ML-based generator achieves orders-of-magnitude lower synthesis time. The ML-approach and verifier are task-model-agnostic and can be readily extended to other scheduling policies or richer task models with constrained deadlines or BCET $\neq$ WCET and LET semantics.

\bibliographystyle{plainurl}
\bibliography{references}

\end{document}